\documentclass[10pt,twocolumn]{IEEEtran}

\usepackage{amsmath,amssymb,amsfonts,amsthm}
\usepackage{graphicx}
\usepackage{subcaption}
\usepackage{algorithm}
\usepackage{algorithmic}
\usepackage[hidelinks]{hyperref}
\usepackage{cite}
\usepackage{textcomp}

\newcommand{\tp}{\normalfont \text{T}}

\newcommand{\Val}{\operatorname{Val}}

\newtheorem{remark}{Remark}

\newtheorem{cor}{Corollary}

\newtheorem{theorem}{Theorem}
\newtheorem{assumption}{Assumption}

\def\qedsymbol{\ensuremath{\Box}}
\def\qed{\ifhmode\unskip\nobreak\fi\quad\qedsymbol}
\def\frqed{\ifhmode\nobreak\hbox to5pt{\hfil}\nobreak%
	\hskip 0pt plus1fill\nobreak\fi\quad\qedsymbol\renewcommand{\qed}{}}
\def\QEDsymbol{\vrule width.6em height.5em depth.1em\relax}
\def\frQED{\ifhmode\nobreak\hbox to5pt{\hfil}\nobreak%
	\hskip 0pt plus1fill\nobreak\fi\quad\QEDsymbol\renewcommand{\qed}{}}
\def\QED{\ifhmode\unskip\nobreak\fi\quad\QEDsymbol}

\begin{document}
\title{\texttt{Flip Co-op}: Cooperative Takeovers in Shared Autonomy}
\author{Sandeep Banik and Naira Hovakimyan
\thanks{This research was supported in part by the National Aeronautics and Space Administration (NASA) under the cooperative agreement 80NSSC20M0229, University Leadership Initiative grant no. 80NSSC22M0070 and Air Force Office of Scientific Research (AFOSR) grant no. FA9550-21-1-0411.}
\thanks{The authors are with the Department of Mechanical Engineering at University of Illinois Urbana-Champaign, Urbana, IL, USA. Emails: \texttt{baniksan@illinois.edu; nhovakim@illinois.edu}.
}}

\maketitle

\begin{abstract}
Shared autonomy requires principled mechanisms for allocating and transferring control between a human and an autonomous agent. Existing approaches often rely on blending control inputs between human and autonomous agent or switching rules, which lack theoretical guarantees. This paper develops a game-theoretic framework for modeling cooperative takeover in shared autonomy. We formulate the switching interaction as a dynamic game in which authority is embedded directly into the system dynamics, resulting in Nash equilibrium(NE)-based strategies rather than ad hoc switching rules. We establish the existence and characterization of NE in the space of pure takeover strategies under stochastic human intent. For the class of linear–quadratic systems, we derive closed-form recursions for takeover strategies and saddle-point value functions, providing analytical insight and efficient computation of cooperative takeover policies. We further introduce a bimatrix potential game reformulation to address scenarios where human and autonomy utilities are not perfectly aligned, yielding a unifying potential function that preserves tractability while capturing intent deviations. The framework is applied to a vehicle trajectory tracking problem, demonstrating how equilibrium takeover strategies adapt across straight and curved path segments. The results highlight the trade-off between human adaptability and autonomous efficiency and illustrate the practical benefits of grounding shared autonomy in cooperative game theory.
\end{abstract}

\section{Introduction}
\label{sec:introduction}
Modern cyber–physical systems (CPS) increasingly rely on \textit{shared autonomy}, where humans and automated agents dynamically share control authority. 
This paradigm is critical in safety-critical domains such as autonomous driving, assistive robotics, and power grid operations, where neither fully manual nor fully automated control is sufficient. 
Autonomous systems excel at rapid computation and optimization but often lack situational awareness or adaptability in novel conditions. 
Conversely, humans bring contextual reasoning and ethical oversight but may be overwhelmed by workload or remain out of the loop during automation~\cite{javdaniSharedAutonomyHindsight2015, nikolaidisHumanRobotMutualAdaptation2017}. 
Failures to coordinate handovers/takeovers have been implicated in real-world incidents from autopilot disengagement in aviation to autonomous vehicle crashes, highlighting the need for principled frameworks for cooperative takeover.

Early approaches to shared autonomy typically combined human and robot commands through control blending or arbitration mechanisms. 
In these methods, the robot predicts the human’s intended goal and then determines the level of assistance to provide. 
A common strategy is to treat the human’s joystick or steering input as a noisy estimate of the intended command, and then compute a weighted combination of the user input and the robot’s autonomous action. 
If the system is confident about the user’s goal, the robot’s action dominates; if not, more weight is placed on the user’s raw command. 
This “predict-then-blend” paradigm has been implemented through Bayesian intent inference, POMDP-based planning, and linear arbitration rules, with applications in assistive teleoperation and robotic manipulation. 
Such methods allow users to complete tasks more efficiently and with less effort, but they rely heavily on the accuracy of the intent prediction model and the choice of blending rule. 
As a result, they often remain task-specific heuristics that lack guarantees of stability, robustness, or generalization across different domains~\cite{jeonSharedAutonomyLearned2020, reddySharedAutonomyDeep2018,javdaniSharedAutonomyHindsight2015}.

Beyond blending, game‑theoretic approaches model shared autonomy as an interaction between two decision‑making agents. 
In these models, the human and robot are treated as players in a game, each optimizing an objective while anticipating the other’s actions. 
Differential‑game formulations have been used for physical collaboration, where equilibrium strategies specify how partners exchange forces and share effort on a joint task~\cite{liDifferentialGameTheory2019}. 
In autonomous driving, interaction has been cast as a dynamic game in which the robot plans while explicitly accounting for how a human driver will respond, enabling not only conflict avoidance but also deliberate influence for safer coordination~\cite{sadighPlanningAutonomousCars2016}. Repeated‑interaction models analyze how humans adapt to robot strategies over time and how robots can optimize behavior in anticipation of such adaptation~\cite{nikolaidisGameTheoreticModelingHuman2017}. 
These works demonstrate that game‑theoretic reasoning captures mutual influence and strategic adaptation compared to control blending. 
However, most formulations still assume symmetric roles or require full knowledge of human utility, which is often unrealistic. 
The present work departs from these assumptions by addressing asymmetric authority and stochastic human intent within a cooperative game-theoretic framework.

Another central challenge in shared autonomy is how control authority is transferred or switched between the human and the autonomous system. 
Unlike blending approaches, which continuously mix inputs, switching models make explicit decisions about who has authority at a given time. 
Prior research explored mode switching in teleoperation, where the system automatically changed control modes or degrees of freedom based on predictions of the user’s intent, reducing cognitive load and clarifying the operator’s goal~\cite{draganPolicyblendingFormalismShared2013}. 
Broader frameworks for adjustable autonomy and mixed-initiative control extended this idea by allowing the system to dynamically allocate decision-making authority depending on task demands, communication delays, or operator workload~\cite{scerriAdjustableAutonomyReal2002}. 
In human–robot teaming, switching has also been studied through adaptive autonomy, where control shifts in real time according to the human’s performance, trust levels, or safety thresholds~\cite{loseyPhysicalInteractionCommunication2022}. 
These approaches emphasize flexibility and situational awareness, but they are typically heuristic and lack general equilibrium analysis. 
The present work differs by embedding switching of authority into a game-theoretic framework, ensuring cooperative equilibria exist and providing principled takeover policies rather than ad hoc rules.

A related body of work on resource takeover games provides the mathematical foundation for this study. The \texttt{FlipDyn} framework introduced the idea of modeling authority over a dynamical system as a competitive game in which two players strategically “flip” control to minimize or maximize long-term performance~\cite{banikFlipDynGameResource2022a}. This model, inspired by the \texttt{FlipIt} game~\cite{vandijkFlipItGameStealthy2013} of stealthy takeovers, formalized how equilibrium takeover strategies emerge when agents contest authority over linear–quadratic systems. Subsequent extensions generalized the framework to networked dynamical systems, capturing how takeovers propagate across interconnected nodes and deriving equilibrium conditions for graph-structured interactions~\cite{zhu2015game}. Other studies explored the role of cyber–physical resilience, applying game-theoretic reasoning to capture adversarial control switching in large-scale infrastructures~\cite{banikFlipDynGraphsResource2025}. Collectively, these works establish rigorous methods for analyzing control handovers in adversarial contexts. The present paper departs from the adversarial setting by translating these mathematical tools to \textit{cooperative shared autonomy}, where both agents pursue aligned objectives but still require principled strategies for when to yield or assert control.

This paper introduces \texttt{Flip Co-op}, a cooperative extension of \texttt{FlipDyn} that models takeover in shared autonomy as a potential game. Unlike adversarial formulations, \texttt{Flip Co-op} assumes a common objective between human and autonomy, aligning their incentives while preserving the asymmetry of human override authority. The framework captures stochastic intent, ensures the existence of Nash equilibria in takeover strategies, and provides closed-form results in linear–quadratic settings. Through this, it offers a principled alternative to blending and heuristic switching, grounding shared autonomy in cooperative game theory with theoretical guarantees. The main contributions are as follows:

\begin{enumerate}
    \item \textbf{Cooperative takeover formulation}:
    We formulate the human–autonomy takeover problem as an identical-interest dynamic game in which the switching variable is embedded directly into the system dynamics. This game encompasses dynamic takeovers, with state-based costs.
   \item \textbf{Existence and characterization of equilibria:}: We establish the existence of Nash equilibrium (NE) in the proposed game and characterize them in the space of pure strategies under stochastic human intent. The analysis yields explicit conditions for equilibrium takeover behavior in general dynamical systems.
    \item \textbf{Analytical solutions for linear–quadratic systems}:  
    For the important class of linear–quadratic (LQ) systems, we derive closed-form recursions for the equilibrium takeover strategies and saddle-point value functions. These analytical results enable efficient computation of cooperative takeover policies in continuous state spaces.
    In particular, we derive closed-form expressions for the NE takeover strategies and parameterized value of the game \emph{independent of the continuous state}.
    \item \textbf{Reformulation as a potential game:}:  
    We introduce a bimatrix potential game representation for settings where human and autonomy utilities are not perfectly aligned. This reformulation identifies a unifying potential function that preserves tractability while capturing deviations in human intent.
\end{enumerate}

We illustrate our results for the scalar, $n-$dimensional systems and vehicle path-tracking problem. 
Section~\ref{sec:Problem_Formulation} formalizes the cooperative takeover problem as an identical-interest dynamic game with switching dynamics. Section~\ref{sec:FlipCoop_General_Systems} establishes the existence and characterization of equilibrium strategies, followed by closed-form analytical results for the linear–quadratic setting in Section~\ref{sec:FlipCoop_Linear_Systems}. Section~\ref{sec:FlipCoop_Potential_games} introduces the potential game representation for partially misaligned utilities, and Section~\ref{sec:Application} demonstrates the framework on a trajectory tracking task. Finally, Section~\ref{sec:Conclusion} summarizes the findings and discusses directions for future work.

\section{Problem Formulation}\label{sec:Problem_Formulation}
Consider a discrete-time dynamical system controlled either by a human or an autonomous agent. The \texttt{FlipDyn} state, $\alpha_{k} \in \{\text{H},\text{A}\}$ indicates whether the human ($\alpha_k = \text{H}$) or the autonomous agent ($\alpha_k = \text{A}$) has control over the system at time $k$. 
The state evolution of the system controlled by the human is given by:
\begin{align}\label{eq:human_dynamics}
	x_{k+1} = F_{k}^{\text{H}}(x_k,u_k),
\end{align}
where $k$ denotes the discrete time index, taking values from the integer set $\mathcal{K} := \{1,2,\dots, L\} \subset \mathbb{N}$, $x_{k} \in \mathbb{R}^{n}$ is the state of the system, $u_{k} \in \mathbb{R}^m$ is the control input of the system, and $F_{k}^{\text{H}}: \mathbb{R}^{n} \times \mathbb{R}^{m} \rightarrow \mathbb{R}^{n}$ is the state transition function. Similarly, the state evolution under an autonomous agent is given by:
\begin{align}\label{eq:autonomous_dynamics}
	x_{k+1} = F_{k}^{\text{A}}(x_k,w_k),
\end{align}
where $F_{k}^{\text{A}}: \mathbb{R}^{n} \times \mathbb{R}^{p} \rightarrow \mathbb{R}^{n}$ is the state transition function under the autonomous agent control input $w_k \in \mathbb{R}^{p}$. 
We describe a takeover through the action $\pi^{j}_k \in \{0,1\}$, which denotes the action of the player $j \in \{\text{H},\text{A}\}$ at time $k$, where $j = \text{H}$ denotes the human and $j=\text{A}$ denotes the autonomous agent. The action $\pi_{k}^{j} = 1, \forall j$ corresponds to takeover/request to takeover and $\pi_{k}^{j} = 0$ otherwise (idle). The binary \texttt{FlipDyn} state updates based on the player's action and prior \texttt{FlipDyn} state. Given $\alpha_{k} = \text{H}$, the \texttt{FlipDyn} state at time $k+1$ is:
\begin{align}\label{eq:flip_state_H}
	\alpha_{k+1} &= \begin{cases}
		\alpha_{k}, & \text{if } \{\pi^{\text{H}}_{k} = 0,  \pi^{\text{A}}_{k} = 0\}, \\
        \text{H}, & \text{if } \{\pi^{\text{H}}_{k} = 0,  \pi^{\text{A}}_{k} = 1\}, \\
        \text{H}, & \text{if } \{\pi^{\text{H}}_{k} = 1,  \pi^{\text{A}}_{k} = 0\}, \\
        \text{A}, & \text{otherwise},
	\end{cases}
\end{align}
where $\pi^{\text{H}}_k = 0$ corresponds to human agent being idle or retaining the control of the system and $\pi^{\text{H}}_k = 1$ represents request to takeover. The first three terms in~\eqref{eq:flip_state_H} correspond to human retaining control of the system and the last term corresponds the request to takeover from the human and autonomous agent taking over. Similarly, given $\alpha_{k} = \text{A}$,  \texttt{FlipDyn} state update is:
\begin{align}\label{eq:flip_state_A}
	\alpha_{k+1} &= \begin{cases}
		\alpha_{k}, & \text{if } \{\pi^{\text{H}}_{k} = 0,  \pi^{\text{A}}_{k} = 0\}, \\
        \text{A}, & \text{if } \{\pi^{\text{H}}_{k} = 0,  \pi^{\text{A}}_{k} = 1\}, \\
        \text{H}, & \text{if } \{\pi^{\text{H}}_{k} = 1,  \pi^{\text{A}}_{k} = 0\} \text{ with probability } p_{k}, \\
        \text{A}, & \text{if } \{\pi^{\text{H}}_{k} = 1,  \pi^{\text{A}}_{k} = 0\} \text{ with probability } 1-p_{k}, \\
        \text{H}, & \text{otherwise},
	\end{cases}
\end{align}
where $\pi^{\text{H}}_k = 1$ corresponds to human agent takeover and $\pi^{\text{A}}_k = 1$ corresponds to request to takeover. The autonomous agent retains control of the system deterministically only when human agent is idle (second condition) and with probability $1-p$ (forth condition) when human agent chooses to takeover. The human agent takes over with probability $p$ when the autonomous does not request to takeover ($\pi^{\text{A}}_k = 0$) and deterministically takes over when both $\pi^{\text{A}}_k = \pi^{\text{H}}_k = 1$. Notice the difference between~\eqref{eq:flip_state_H} and~\eqref{eq:flip_state_A} is the human agent being uncertain with probability $p$ and ability to takeover the system irrespective of the autonomous agent's actions. Such a model captures the current design of shared control system where the human agent has higher authority over autonomous agents. 
The hybrid \texttt{FlipDyn} dynamics is compactly described by:
\begin{align}
	\label{eq:FlipDyn_compact}
	\alpha_{k+1} &= \begin{cases}
	                  \begin{aligned}
                            & \bar{\pi}^{\text{H}}_k\bar{\pi}^{\text{A}}_k\alpha_{k} + \pi_{k}^{\text{A}}\pi_{k}^{\text{H}}\text{A}\\ & +  \left(\bar{\pi_{k}}^\text{A}\pi_{k}^{\text{H}} + \pi_{k}^{\text{A}}{\bar{\pi_{k}}}^\text{H}\right)\text{H }, 
	                    \end{aligned} & \text{if } \alpha_{k} = \text{H}, \\[10pt] 
                        \begin{aligned}
                            & \bar{\pi}^{\text{H}}_k\bar{\pi}^{\text{A}}_k\alpha_{k} +  \pi_{k}^{\text{A}}{\bar{\pi_{k}}}^\text{H}\text{A} \\ &  + \pi_{k}^{\text{H}}\text{H},
	                    \end{aligned} & \begin{aligned}
	                        & \text{If } \alpha_{k} = \text{A} \\ &
                                \text{with probability } p_{k}, 
	                    \end{aligned} \\[10 pt] 
                        \begin{aligned}
                            & \bar{\pi}^{\text{H}}_k\bar{\pi}^{\text{A}}_k\alpha_{k} + \pi_{k}^{\text{A}}\pi_{k}^{\text{H}}\text{H}\\ & +  \left({\bar{\pi_{k}}}^\text{A}\pi_{k}^{\text{H}} + \pi_{k}^{\text{A}}{\bar{\pi_{k}}}^\text{H}\right)\text{A } 
	                    \end{aligned}, & \begin{aligned}
	                        & \text{ }  \\ &
                                \text{with probability } 1-{p}_{k}, 
	                    \end{aligned}
                	\end{cases} 
\end{align}
where $\bar{x} := 1 - x$ is a binary variable. Takeovers are mutually exclusive, i.e., only one player is in control of the system at any given time. The continuous state $x_{k+1}$ at time $k+1$ is dependent on $\alpha_{k+1}$. 
In this work, we aim to solve for a takeover strategy for both the human and autonomous agent. Given a non-zero initial state $x_{1}$, we pose the takeover problem as a identical interest dynamic game described by the dynamics~\eqref{eq:human_dynamics}~\eqref{eq:autonomous_dynamics} and~\eqref{eq:FlipDyn_compact} over a finite-horizon $L$, where both the human and autonomous agent aims to minimize a net cost given by:
\begin{equation}\label{eq:obj_def_alpha}
	\begin{aligned}
		J(x_{1}, \alpha_{1}, \{\pi^{\text{H}}_{\mathbf{L}}\}, \{\pi^{\text{A}}_{\mathbf{L}}\}) & = g_{L+1}(x_{L+1}, \alpha_{L+1}) + \sum_{t=1}^{L} g_t(x_t, \alpha_t) \\  & + \pi_{t}^{\text{H}}h_t(x_t) + \pi^{\text{A}}_{t}a_t(x_t), 
	\end{aligned}
\end{equation}
where $g_t(x_t, \alpha_t): \mathbb{R}^{n} \times \{\text{H},\text{A}\} \rightarrow \mathbb{R}$ represents the state cost with $g_{L+1}(x_{L+1}, \alpha_{L+1}): \mathbb{R}^{n} \times \{\text{H},\text{A}\} \rightarrow \mathbb{R}$ representing the terminal state cost, $h_t(x_t): \mathbb{R}^{n} \rightarrow \mathbb{R}$  and $a_t(x_t):\mathbb{R}^{n} \rightarrow \mathbb{R}$ are the instantaneous takeover costs for the human and autonomous agent, respectively. Let $\{\pi_{\mathbf{L}}^j\} := \{\pi_1^j, \dots, \pi_{L}^j\}, j \in \{\text{H},\text{A}\}$. We pose the problem~\eqref{eq:obj_def_alpha} as a dynamic game between the human and autonomous agent, termed as the \emph{\texttt{Flip Co-op}}, where both players aim to optimize their own takeover strategies. In particular, we consider a subclass of games known as identical interest games~\cite{hespanha2017noncooperative}, where both agents are aligned in minimizing a common cost function~\eqref{eq:obj_def_alpha}. 

\smallskip

We seek to find Nash Equilibrium (NE) solutions of the game~\eqref{eq:obj_def_alpha}. To guarantee the existence of a pure or mixed NE takeover strategy, we expand the set of player policies to behavioral strategies --  probability distributions over the space of discrete actions at each time step~\cite{hespanha2017noncooperative}.
Specifically, let
\begin{equation}\label{eq:behav_pol}
	y_{k}^{\alpha_{k}} = \begin{bmatrix}
		1 - \beta_{k}^{\alpha_{k}} & \beta_{k}^{\alpha_{k}}
	\end{bmatrix}^{\text{T}} \text{ and \hfill} z_{k}^{\alpha_{k}} = \begin{bmatrix}
	1 - \gamma_{k}^{\alpha_{k}} & \gamma_{k}^{\alpha_{k}}
\end{bmatrix}^{\text{T}},
\end{equation}
be the behavioral strategies for the human and autonomous agent at time instant $k$ for the \texttt{FlipDyn} state $\alpha_k$, such that $\beta_{k}^{\alpha_{k}} \in [0,1]$ and $\gamma_{k}^{\alpha_{k}} \in [0,1]$, respectively. The takeover actions
\[
\pi_{k}^{\text{H}} \sim y_{k}^{\alpha_{k}}, \quad  \pi_{k}^{\text{A}} \sim z_{k}^{\alpha_{k}} 
\]
of each player at any time $k$ are sampled from the corresponding behavioral strategy. The behavioral strategies are given by $y_{k}^{\alpha_{k}}, z_{k}^{\alpha_{k}}  \in \Delta_{2}$, where $\Delta_{2}$ is the probability simplex in two dimensions. 
Over the finite-horizon $L$, let $y_{\mathbf{L}} := \{y_{1}^{\alpha_{1}}, y_{2}^{\alpha_{2}}, \dots, y_{L}^{\alpha_{L}}\} \in \Delta^{L}_{2}$ and $z_{\mathbf{L}} := \{z_{1}^{\alpha_{1}}, z_{2}^{\alpha_{2}}, \dots, z_{L}^{\alpha_{L}}\} \in \Delta^{L}_{2}$ be the sequence of human and autonomous behavioral strategies. Thus, the expected outcome of the game~\eqref{eq:obj_def_alpha} is given by:
\begin{equation}
    \label{eq:opti_E_cost}
	J_{E}(x_1, \alpha_{1}, y_{\mathbf{L}}, z_{\mathbf{L}}) :=  \mathbb{E}[J( x_{1}, \alpha_{1}, \{\pi^{\text{H}}_{L}\}, \{\pi^{\text{A}}_{L}\})],
\end{equation}
where the expectation is evaluated with respect to the distributions $y_{\mathbf{L}}$ and $z_{\mathbf{L}}$. Specifically, we seek a saddle-point solution ($y_{\mathbf{L}}^{*}, z_{\mathbf{L}}^{*}$) in the space of behavioral strategies such that for any non-zero initial state $x_0 \in \mathbb{R}^{n}, \alpha_0 \in \{0,1\}$,
\begin{equation*}
    \begin{aligned}
         J_E(x_0, \alpha_0, y_{\mathbf{L}}^{*}, z_{\mathbf{L}}^{*}) \leq J_E(x_0, \alpha_0, y_{\mathbf{L}}, z_{\mathbf{L}}^{*}) \leq J_E(x_0, \alpha_0, y_{\mathbf{L}}^{*}, z_{\mathbf{L}}).
    \end{aligned}
\end{equation*}
The \texttt{Flip Co-op} game is completely defined by the expected cost~\eqref{eq:opti_E_cost} and the space of player takeover strategies subject to the dynamics~\eqref{eq:human_dynamics},~\eqref{eq:autonomous_dynamics} and \eqref{eq:FlipDyn_compact}. In the next section, we derive the outcome of the \texttt{Flip Co-op} game for both the \texttt{FlipDyn} state of $\alpha = \text{H}$ and $\alpha = \text{A}$ for general systems.

\section{\texttt{Flip Co-op} for general systems}\label{sec:FlipCoop_General_Systems}
We will begin by deriving the NE takeover strategies of the \texttt{Flip Co-op} game, given any control policy pair  $u_{\mathbf{L}}, w_{\mathbf{L}}$, in each of the two takeover. Our approach begins by defining the saddle-point value of the game.
\subsection{Saddle-point value}
Given a \texttt{FlipDyn} state at time $k \in \mathcal{K}$, the saddle-point value comprises  an instantaneous state cost and an additive cost-to-go based on the players takeover actions. The cost-to-go is determined via a cost-to-go matrix in each \texttt{FlipDyn} state $\alpha_{k} = \text{H}$ and $\alpha_{k} = \text{A}$, represented by $\Xi_{k+1}^{\text{H}} \in \mathbb{R}^{2 \times 2}$ and $\Xi_{k+1}^{\text{A}} \in \mathbb{R}^{2 \times 2}$, respectively. Let $V^{\text{H}}_k(x, \Xi_{k+1}^{\text{H}})$ and $V^{\text{A}}_k(x, \Xi_{k+1}^{\text{A}})$ be the saddle-point values as a function of the continuous state $x$ and cost-to-go matrices for $\alpha_k = \text{H}$ and $\alpha_k = \text{A}$, respectively. The entries of the cost-to-go matrix $\Xi^{\text{H}}_{k+1}$ corresponding to each pair of takeover actions are given by:
\begin{equation}\label{eq:Cost_to_go_H}
    \begin{aligned}
		& \begin{matrix} & \hphantom{000} \text{Idle} & & \hphantom{v_{k+1}^0(.,.} \text{Takeover}\end{matrix} \\
		\begin{matrix} \text{Idle} \\[5pt] \text{Request to}\\\text{takeover} \end{matrix} & \underbrace{\begin{bmatrix}
			v_{k+1}^{\text{H}} &  v_{k+1}^{\text{H}} + a_k(x)  \\[8pt]
			v_{k+1}^{\text{H}} + h_k(x) &  v_{k+1}^{\text{A}} + h_k(x) + a_k(x) \\[5pt] 
		\end{bmatrix}}_{\Xi_{k+1}^{\text{H}}}
	\end{aligned},
\end{equation}
\begin{align}
    \label{eq:V_k_0} \text{where } \ & v_{k+1}^{\text{H}} := V_{k+1}^{\text{H}}\left(F_k^{\text{H}}(x,u_k),\Xi_{k+2}^{\text{H}}\right), \\
    \label{eq:V_k_1} & v_{k+1}^{\text{A}} := V_{k+1}^{\text{A}}(F_k^{\text{A}}(x,w_k),\Xi_{k+2}^{\text{A}}).
\end{align}
The row and column entries of $\Xi^{\text{H}}_{k+1}$ are based on the player's actions, described by~\eqref{eq:FlipDyn_compact} and its associated dynamics~\eqref{eq:human_dynamics},~\eqref{eq:autonomous_dynamics}. We will refer $X(i,j)$ as the $(i,j)$-th entry of the matrix $X$. The first row entries $\Xi_{k+1}^{\text{H}}(1,1)$ and $\Xi_{k+1}^{\text{H}}(1,2)$ corresponds to the human agent remaining idle, which prevents the autonomous agent to takeover despite the column action of takeover. The second row entries $\Xi_{k+1}^{\text{H}}(2,1)$ and $\Xi_{k+1}^{\text{H}}(2,2)$ corresponds to action of request to takeover, and transition to the autonomous agent ($v^{\text{A}}_{k+1}$) only when the autonomous agent agrees (column action of takeover). The entries of $\Xi_{k+1}^{\text{H}}$ couples the saddle-point value in each \texttt{FlipDyn} state. 
At time $k$ for a given human agent control policy $u_k$, state $x$ and $\alpha_k=\text{H}$, the saddle-point value satisfies
\begin{equation}
    \label{eq:V_k^0_cost_to_go}
	V^{\text{H}}_k(x, \Xi_{k+1}^{\text{H}}) = g_k(x,\text{H})  + \Val(\Xi^{\text{H}}_{k+1}), 
\end{equation}
where $\Val(X_{k+1}^{\alpha_{k}}):= \min_{y_{k}^{\alpha_{k}}} \min_{z_{k}^{\alpha_{k}}} y_{k}^{{\alpha_{k}}^{\tp}}X_{k+1}z_{k}^{\alpha_{k}}$, represents the saddle-point value of the identical matrix $X_{k+1}$ for the \texttt{FlipDyn} state $\alpha_{k}$, and $\Xi^{0}_{k+1} \in \mathbb{R}^{2 \times 2}$ is the cost-to-go zero-sum matrix. The human's (row player) and autonomous agent's (column player) action results in either an entry within $\Xi^{\text{H}}_{k+1}$ (if the matrix has a saddle point in pure strategies) or in the expected sense, resulting in a cost-to-go from state $x$ at time $k$.

Similarly, for $\alpha_k = \text{A}, \forall k$, the cost-to-go matrix entries $\Xi_{k+1}^{\text{A}}$ are:
\begin{equation}\label{eq:Cost_to_go_A}
    \begin{aligned}
		& \begin{matrix} & \hphantom{0000} \text{Idle} & & \hphantom{v_{k+1}^0000} \text{Request to takeover}\end{matrix} \\
		\begin{matrix} \text{Idle} \\[10pt] \text{Takeover}\end{matrix} & \underbrace{\begin{bmatrix}
			v_{k+1}^{\text{A}} \hphantom{000}  &  \hphantom{00} v_{k+1}^{\text{A}} + a_k(x) \hphantom{0}  \\[5pt]
			\begin{matrix}
                    p_{k}v_{k+1}^{\text{H}} + (1-p_{k})v_{k+1}^{\text{A}} \\
                        + h_k(x)       
                \end{matrix}
              &  \begin{matrix}
                    v_{k+1}^{\text{H}} + h_k(x) \\
                        + a_k(x)       
                \end{matrix}
		\end{bmatrix}}_{\Xi_{k+1}^{\text{A}}}.
	\end{aligned}
\end{equation}
The first row entries of $\Xi_{k+1}^{\text{A}}$ corresponds to autonomous agent retaining control of the system as the human agent remains idle. The second row entry $\Xi_{k+1}^{\text{A}}(2,1)$ corresponds to the autonomous agent remaining idle, but the human agent can takeover the system with a probability $p$ and the system remaining in control of the autonomous agent with probability $1-p$. Finally, the entry $\Xi_{k+1}^{\text{A}}(2,2)$ corresponds to takeover of the system by the human agent. Analogous to~\eqref{eq:V_k^0_cost_to_go} the saddle-point value for $\alpha_{k} = \text{A}$ satisfies:
\begin{equation}
    \label{eq:V_k^1_cost_to_go}
	\text{with} \ V^{\text{A}}_k(x, \Xi_{k+1}^{1}) = g_k(x,\text{A}) +  \Val(\Xi^{\text{A}}_{k+1}).
\end{equation}

With the saddle-point values established in each of the \texttt{FlipDyn} states, in the following subsection, we will characterize the NE takeover strategies and the saddle-point values over the time-horizon $L$.

\subsection{NE takeover strategies of the \texttt{Flip Co-op} game}
In order to characterize the saddle-point value of the game, we restrict the cost functions to belong to a finite domain, stated in the following assumption.
\begin{assumption}\label{ast:general_costs}
    [Non-negative costs] At any time instant $k \in \mathcal{K}$, the state and takeover costs $g_k(x,\alpha), h_k(x), a_k(x)$ for all $x \in \mathbb{R}^{n},$ and $\alpha \in \{\text{H},\text{A}\}$ are non-negative $(\mathbb{R}_{\geq 0})$. 
\end{assumption}

Assumption~\ref{ast:general_costs} ensures that cost comparisons within the cost-to-go matrix are sign-consistent, allowing for a clear characterization of player strategies, whether pure or mixed.
Under Assumption~\ref{ast:general_costs}, we derive the saddle-point value and takeover policies for the finite time-horizon in the following result.

\begin{theorem}\label{th:NE_Val_gen_FDC}
    (Case $\alpha_k = \text{H}$) Under Assumption~\ref{ast:general_costs} and known control policies, $u_{\mathbf{L}}$ and $w_{\mathbf{L}}$, the unique NE takeover strategies of the \texttt{Flip Co-op} game~\eqref{eq:opti_E_cost} at time $k \in \mathcal{K}$, subject to the continuous state dynamics~\eqref{eq:human_dynamics},~\eqref{eq:autonomous_dynamics} and \texttt{FlipDyn} dynamics~\eqref{eq:FlipDyn_compact} are given by:
    \begin{align}
    \begin{split}\label{eq:TP_gen_H}
            \{\beta^{\text{H}*}_{k}, \gamma^{\text{H}*}_{k}\} = \begin{cases}
            \{0,0\}, & \text{if } v^{\text{H}}_{k+1} < \begin{matrix}
                v^{\text{A}}_{k+1} + h_{k}(x) + a_{k}(x),
            \end{matrix} \\[5 pt]
            \{1,1\}, & \text{otherwise}.
            \end{cases} 
    \end{split}
    \end{align}
    
    The saddle-point value is given by:
    \begin{align}\label{eq:Val_gen_H}
        v_{k}^{\text{H}} = 
        \begin{cases}
            \begin{aligned}
				& g_k(x,\text{H}) + v_{k+1}^{\text{H}},
            \end{aligned} &\text{if } v^{\text{H}}_{k+1} < \begin{matrix}
                v^{\text{A}}_{k+1} + h_{k}(x) \\ + a_{k}(x),
            \end{matrix} \\
            \begin{aligned}
				& g_k(x,\text{H}) + v_{k+1}^{\text{A}} \\ & + h_{k}(x) + a_{k}(x),
            \end{aligned} &\text{otherwise},
		\end{cases} 
	\end{align}
    where $v_{k}^{\text{H}}:= V_{k}^{\text{H}}(x,\Xi_{k+1}^{\text{H}})$ and 
    $v_{k}^{\text{A}}:= V_{k}^{\text{A}}(x,\Xi_{k+1}^{\text{A}})$.
    \medskip
    
    (Case $\alpha_k = \text{A}$) The unique NE takeover strategies are 
    \begin{align}
    \begin{split}\label{eq:TP_gen_A}
            \{\beta_{k}^{\text{A}*},\gamma_{k}^{\text{A}*}\}  = \begin{cases}
            \{0,0\}, & \text{if } \ 
                \begin{matrix}
                    v^{\text{A}}_{k+1} - v^{\text{H}}_{k+1} < \dfrac{h_{k}(x)}{p_{k}}
                \end{matrix}, \\[10pt]
            \{1,0\}, & \text{if } \ 
                \begin{matrix}
                    v^{\text{A}}_{k+1} - v^{\text{H}}_{k+1} \leq \dfrac{a_{k}(x)}{1 - p_{k}}  \\
                    v^{\text{A}}_{k+1} - v^{\text{H}}_{k+1} \geq \dfrac{h_{k}(x)}{p_{k}},
                \end{matrix} \\
            \{1,1\}, & \text{otherwise.}
            \end{cases} 
    \end{split}
    \end{align}
    The saddle-point value is given by:
    \begin{align}\label{eq:Val_gen_A}
        v_{k}^{\text{A}} = 
        \begin{cases}
            \begin{aligned}
				& g_k(x,\text{A}) + v_{k+1}^{\text{A}},
            \end{aligned} &\text{if } \begin{matrix}
                    v^{\text{A}}_{k+1} - v^{\text{H}}_{k+1} < \dfrac{h_{k}(x)}{p_{k}}
                \end{matrix}, \\[10pt]
            \begin{aligned}
				& g_k(x,\text{A}) + p_{k}v^{\text{H}}_{k+1} + \\[5pt] & (1 - p_{k})v^{\text{A}}_{k+1} + h_k(x),
            \end{aligned} &\text{if } \begin{matrix}
                    v^{\text{A}}_{k+1} - v^{\text{H}}_{k+1} \leq \dfrac{a_{k}(x)}{1 - p_{k}}  \\
                    v^{\text{A}}_{k+1} - v^{\text{H}}_{k+1} \geq \dfrac{h_{k}(x)}{p_{k}},
                \end{matrix} \\
            \begin{aligned}
                g_k(x,\text{A}) + v_{k+1}^{\text{H}} + \\
                h_{k}(x) + a_{k}(x),
            \end{aligned}
             & \text{otherwise}.
		\end{cases} 
	\end{align}
    The boundary condition at $k = L$ is given by:
    \begin{gather}\label{eq:b_cond_NE_Val_gen_FDC}
        \Xi_{L+2}^{\text{H}} := \mathbf{0}_{2 \times 2}, \ \Xi_{L+2}^{\text{A}} := \mathbf{0}_{2 \times 2}, 
    \end{gather}
    where $\mathbf{0}_{i \times j} \in \mathbb{R}^{i \times j}$ represents a matrix of zeros.
    \frqed
\end{theorem}

\begin{proof}
    We will only derive the NE takeover strategies and saddle-point value for case of $\alpha_k = \text{A}$. We leave out the derivations for $\alpha = \text{H}$ as they are analogous to $\alpha = \text{A}$. There are three pure Nash equilibrium policies to consider for the $2 \times 2$ matrix game defined by the matrix in~\eqref{eq:Cost_to_go_A}.
    \smallskip
    
    i) \underline{Pure strategy - \{Idle, Idle\}}:
    Both the human and autonomous agent choose the action of staying idle. 
    \noindent We determine the conditions under which such a pure policy is feasible. Under Assumption~\ref{ast:general_costs}, we compare the entries of $\Xi_{k+1}^{\text{A}}$ when the human agent opts to remain idle to obtain the condition: 
    \begin{equation}
        \begin{aligned}
            v_{k+1}^{\text{A}} < v_{k+1}^{\text{A}} + a_{k}(x),
        \end{aligned}
    \end{equation}
    which indicates that the autonomous agent always chooses to play idle. Next, we compare the entry of $ v_{k+1}^{\text{A}}$ against a human agent takeover action to obtain:
    \begin{equation*}
        \begin{aligned}
            v_{k+1}^{\text{A}} & \leq p_{k}v_{k+1}^{\text{H}} + (1 - p_{k})v_{k+1}^{\text{A}} + h_{k}(x), 
        \end{aligned}
    \end{equation*}
    \begin{equation*}
        \begin{aligned}
         \Rightarrow & v_{k+1}^{\text{A}} - v_{k+1}^{\text{H}} \leq \dfrac{h_k(x)}{p_{k}}.
        \end{aligned}
    \end{equation*}
    Notice that such a strategy might be a non-admissible NE. 
    The saddle-point value corresponding to the pure strategy of playing idle by both the players, is the entry $\Xi_{k+1}^{\text{A}}(1,1)$, given by:
    \begin{equation*}
        V_{k}^{\text{A}}(x,\Xi_{k+1}^{\text{A}}) = g_k(x,\text{A}) + v_{k+1}^{\text{A}}.
    \end{equation*}
    
    ii) \underline{Pure strategy - \{Takeover, Idle\}}:
    The autonomous agent chooses to stay idle whereas the human agent chooses to takeover. 
    \noindent To realize such a pure strategy, we compare the entry $\Xi_{k+1}^{\text{A}}(2,1)$ against $\Xi_{k+1}^{\text{A}}(1,1)$ to obtain the condition:
    \begin{equation*}
        \begin{aligned}
            v_{k+1}^{\text{A}} - v_{k+1}^{\text{H}} \geq \dfrac{h_{k}(x)}{p_{k}}.
        \end{aligned}
    \end{equation*}
    Similarly, comparing $\Xi_{k+1}^{\text{A}}(2,1)$ against $\Xi_{k+1}^{\text{A}}(2,2)$ yields:
    \begin{equation*}
        \begin{aligned}
            p_{k}v_{k+1}^{\text{H}} + (1 - p_{k})v_{k+1}^{\text{A}} + h_{k}(x) & \leq v_{k+1}^{\text{H}} + h_{k}(x) + a_{k}(x), \\
            \Rightarrow v_{k+1}^{\text{A}} - v_{k+1}^{\text{H}} & \leq \dfrac{a_{k}(x)}{1 - p_{k}}.
        \end{aligned}
    \end{equation*}
    If the conditions derived for the strategy are satisfied, then such a strategy corresponds to an admissible NE. 
    The expected saddle-point corresponds to the entry $\Xi_{k+1}^{\text{A}}(1,2)$, given by:
    \begin{equation*}
        V_{k}^{\text{A}}(x,\Xi_{k+1}^{\text{A}}) = g_k(x,\text{A}) + p_{k}v_{k+1}^{\text{H}} + (1 - p_{k})v_{k+1}^{\text{A}} + h_{k}(x).
    \end{equation*}

    iii) \underline{Pure strategy - \{Takeover, Request to takeover\}}: The human agent takes over and the autonomous agent requests to takeover. Following the same process, we compare the entry $\Xi_{k+1}^{\text{A}}(2,2)$ against $\Xi_{k+1}^{\text{A}}(2,1)$ to obtain the condition:
    \begin{equation*}
        v_{k+1}^{\text{A}} - v_{k+1}^{\text{H}} > \frac{a_{k}(x)}{1 - p_{k}}.
    \end{equation*}
    Notice, we didn't compare $\Xi_{k+1}^{\text{A}}(2,2)$ to $\Xi_{k+1}^{\text{A}}(1,2)$, because the entry $\Xi_{k+1}^{\text{A}}(1,2)$ is not a NE. In other words, there is no incentive for any player to opt a policy which corresponds to the entry $\Xi_{k+1}^{\text{A}}(1,2)$, since $\Xi_{k+1}^{\text{A}}(1,1)$ is strictly better than $\Xi_{k+1}^{\text{A}}(1,2)$. 
    The saddle-point value for the derived policy is:
    \begin{equation*}
        V_{k}^{\text{A}}(x,\Xi_{k+1}^{\text{A}}) = g_{k}(x, \text{A}) + v_{k+1}^{\text{H}} + h_{k}(x) + a_{k}(x).
    \end{equation*}
    Collecting all the saddle-point values of the game for the derived pure strategy NE, we obtain the saddle-point value update equation over the  horizon of $L$ in~\eqref{eq:Val_gen_A}. The boundary conditions~\eqref{eq:b_cond_NE_Val_gen_FDC} imply that the saddle-point values at $k = L+1$ satisfy
    \begin{equation*}
        \begin{aligned}
            V_{L+1}^{\text{H}}(x,\mathbf{0}_{2 \times 2}) & = g_{L+1}^{\text{H}}(x,\text{H}), \\ V_{L+1}^{1}(x,\mathbf{0}_{2 \times 2}) & = g_{L+1}^{\text{A}}(x,\text{A}).
        \end{aligned}
    \end{equation*}
\end{proof}

For a finite cardinality of the state $\mathcal{X}$, fixed player policies $u_k$ and $w_k,  k \in \mathcal{K}$, and a finite-horizon $L$, Theorem~\ref{th:NE_Val_gen_FDC} yields an exact saddle-point value of the \texttt{Flip Co-op} game~\eqref{eq:opti_E_cost}. However, the computational and storage complexities  scale undesirably with the cardinality of $\mathcal{X}$, especially in continuous state spaces. For this purpose, in the next section, we will provide a parametric form of the saddle-point value for the case of linear dynamics with quadratic costs.

\section{\texttt{Flip Co-op} for LQ Problems}\label{sec:FlipCoop_Linear_Systems}
To address continuous state spaces arising in the \texttt{Flip Co-op} game, we restrict our attention to a linear dynamical system with quadratic costs (LQ problems). The dynamics of a linear system at time instant $k \in \mathcal{K}$, controlled by the human agent satisfies:
\begin{equation}
    \label{eq:H_control_dynamics}
    \begin{aligned}
        x_{k+1} & = F_{k}^{\text{H}}(x_k, u_k) := E_kx_k + B_ku_k,
    \end{aligned}
\end{equation}
where $E_{k} \in \mathbb{R}^{n \times n}$ denotes the state transition matrix, while $B_{k} \in \mathbb{R}^{n \times m}$ represents the human agent control matrix. Similarly, the dynamics of the linear system when the autonomous agent takes over satisfies:
\begin{equation}
    \label{eq:A_control_dynamics}
    \begin{aligned}
        x_{k+1} & = F_{k}^{\text{A}}(x_k, w_k) := E_kx_k + C_kw_k,
    \end{aligned}
\end{equation}
where $C_{k} \in \mathbb{R}^{n \times p}$ signifies the autonomous agent control matrix.
The \texttt{FlipDyn} dynamics~\eqref{eq:FlipDyn_compact} then reduces to:
\begin{align}\label{eq:linear_dynamics}
	x_{k+1} = E_{k}x_{k} + \mathbf{1}_{\text{H}}(\alpha_{k+1})B_{k}u_{k} + \mathbf{1}_{\text{A}}(\alpha_{k+1})H_{k}w_{k},
\end{align}
where $\mathbf{1}_{\mathcal{A}}: \alpha_{k+1} \to \{0,1\}$ is an indicator function, which maps to one if $\alpha_{k+1} = \mathcal{A}$ and zero otherwise.
% We use the following assumed costs.
The stage and takeover quadratic costs are:
\begin{gather}\label{eq:cost_quad}
        g_k(x,\alpha_k) = x^{\tp}G_k^{\alpha_k}x, \\
        h_k(x) = x^{\tp}H_kx, \ a_k(x) = x^{\tp}A_kx, \nonumber
\end{gather}
where $G_k^{\alpha_k} \in \mathbb{S}^{n \times n}_{+}, H_k \in \mathbb{S}^{n \times n}_{+}, A_k \in \mathbb{S}^{n \times n}_{+}$ are positive definite matrices. 

\medskip

\begin{remark}
The control policies for both players act mutually exclusive within their respective  \texttt{FlipDyn} state. Specifically, the human control policy $u_k$ affects the dynamics when the \texttt{FlipDyn} state $\alpha_k = \text{H}$, while the autonomous agent control policy $w_k$ comes into effect when $\alpha_k = \text{A}$.
\end{remark}

If we constrain the control polices of both players to function of the continuous state $x$, then the saddle-point value for each \texttt{FlipDyn} state can be represented purely as a function of the continuous state $x$ and \texttt{FlipDyn} state, in contrast to being contingent on both continuous state $x$ and the control input for each \texttt{FlipDyn} state. This restriction is formally stated in the subsequent assumption.

\medskip
\begin{assumption}\label{ast:linear_control_space}
    We restrict the control policies to be linear state-feedback in the continuous state $x$, described as:
    \begin{equation}\label{eq:linear_FD_control}
        u_k(x) := K_kx, \quad w_k(x) := W_kx,
    \end{equation}
    where $K_k \in \mathbb{R}^{m \times n}$ and $W_k \in \mathbb{R}^{p \times n}$ are human and autonomous agent control gains matrices, respectively. 
\end{assumption}

Under Assumption~\ref{ast:linear_control_space}, the human and autonomous agent dynamics can be compactly written as:
\begin{subequations}\label{eq:dynamics_HA_compact}
    \begin{align}
        x_{k+1} = \tilde{B}_{k}x_{k} := (E_{k} + B_{k}K_{k})x_{k}, \\
        x_{k+1} = \tilde{C}_{k}x_{k} := (E_{k} + C_{k}W_{k})x_{k}.
    \end{align}
\end{subequations}
Given the linear dynamics~\eqref{eq:dynamics_HA_compact}, we postulate a parametric form for the saddle-point value in each \texttt{FlipDyn} state as follows:
\begin{equation}
    \label{eq:para_form_al_QC}
    \begin{aligned}
        & V_{k}^{\text{H}}(x,\Xi_{k+1}^{\text{H}}) \Rightarrow V_{k}^{\text{H}}(x) := x^{\tp}P^{\text{H}}_kx, \\
        & V_{k}^{\text{A}}(x, \Xi_{k+1}^{\text{A}}) \Rightarrow V_{k}^{\text{A}}(x) := x^{\tp}P^{\text{A}}_kx,
    \end{aligned}
\end{equation}
where $P^{\text{H}}_{k} \in \mathbb{S}_{+}^{n \times n}$ and $P^{\text{A}}_{k}\in \mathbb{S}_{+}^{n \times n}$ correspond to the \texttt{FlipDyn} states $\alpha = \text{H}$ and $\text{A}$, respectively. 
We impose Assumption~\ref{ast:linear_control_space} to enable factoring out the state $x$ while computing the saddle-point value update backward in time. 

Next, we outline the NE takeover strategies of both the players, along with the corresponding saddle-point values for each of the \texttt{FlipDyn} state.

\begin{cor}\label{cor:NE_Val_FDC_H}
    (Case $\alpha_k = \text{H}$) The unique NE takeover strategies of the \texttt{FlipDyn Co-op} game~\eqref{eq:opti_E_cost} for every $k \in \mathcal{K}$, subject to the dynamics~\eqref{eq:dynamics_HA_compact}, with quadratic and takeover costs~\eqref{eq:cost_quad} and \texttt{FlipDyn} dynamics~\eqref{eq:FlipDyn_compact} are given by:
    \begin{align}\label{eq:TP_quadcost_H}
            \{\beta^{\text{H}*}_{k}, \gamma^{\text{H}*}_{k}\} = \begin{cases}
            \{0,0\}, & \text{if } x^{\tp}\tilde{P}_{k+1}x < \begin{matrix}
                x^{\tp}(H_{k} + A_{k})x,
            \end{matrix} \\[3 pt]
            \{1,1\}, & \text{otherwise}
            \end{cases} 
    \end{align}
    \begin{align}\label{eq:Val_quadcost_H}
        P_{k}^{\text{H}} = 
        \begin{cases}
            \begin{aligned}
				& G_k^{\text{H}} + \tilde{B}_{k}^{\tp}P_{k+1}^{\text{H}}\tilde{B}_{k},
            \end{aligned} &\text{if } \tilde{P}_{k+1} \prec \begin{matrix}
                 H_{k} + A_{k},
            \end{matrix} \\[5pt]
            \begin{aligned}
				& G_k^{\text{H}} + \tilde{C}_{k}^{\tp}P_{k+1}^{\text{A}}\tilde{C}_{k} \\ & + H_{k} + A_{k},
            \end{aligned} &\text{otherwise},
		\end{cases} 
    \end{align}
    where $\tilde{P}_{k+1} := \tilde{B}_{k}^{\tp}P_{k+1}^{\text{H}}\tilde{B}_{k} - \tilde{C}_{k}^{\tp}P_{k+1}^{\text{A}}\tilde{C}_{k}$.
    
    \noindent (Case $\alpha_k = \text{A}$) The unique NE takeover strategies are given by:
    \begin{align}
    \begin{split}\label{eq:TP_quadcost_A}
            \{\beta_{k}^{\text{A}*},\gamma_{k}^{\text{A}*}\}  = \begin{cases}
            \{0,0\}, & \text{if } \ 
                \begin{matrix}
                    x^{\tp}\tilde{P}_{k+1}x > \dfrac{x^{\tp}H_{k}x}{p_{k}}
                \end{matrix}, \\[10pt]
            \{1,0\}, & \text{if } \ 
                \begin{matrix}
                    x^{\tp}\tilde{P}_{k+1}x \geq \dfrac{x^{\tp}A_{k}x}{1 - p_{k}}  \\[5pt]
                    x^{\tp}\tilde{P}_{k+1}x \geq \dfrac{x^{\tp}H_{k}x}{p_{k}},
                \end{matrix} \\
            \{1,1\}, & \text{otherwise.}
            \end{cases} 
    \end{split}
    \end{align}
    The saddle-point value is given by:
    \begin{align}\label{eq:Val_quadcost_A}
        P_{k}^{\text{A}} = 
        \begin{cases}
            \begin{aligned}
				& G_k^{\text{A}} + \tilde{C}_{k}^{\tp}P_{k+1}^{\text{A}}\tilde{C}_{k},
            \end{aligned} &\text{if } \begin{matrix}
                    \tilde{P}_{k+1} \succ \dfrac{H_{k}}{p_{k}}
                \end{matrix}, \\[10pt]
            \begin{aligned}
				& G_k^{\text{A}} + p_{k}\tilde{B}_{k}^{\tp}P^{\text{H}}_{k+1}\tilde{C}_{k} + H_{k}  \\[5pt] & + (1 - p_{k})\tilde{C}_{k}^{\tp}P^{\text{A}}_{k+1}\tilde{C}_{k},
            \end{aligned} &\text{if } \begin{matrix}
                    \tilde{P}_{k+1} \succeq  \dfrac{A_{k}}{1 - p_{k}}  \\
                    \tilde{P}_{k+1} \preceq \dfrac{H_{k}}{p_{k}},
                \end{matrix} \\
            \begin{aligned}
                & G_k^{\text{A}} + \tilde{B}_{k}^{\tp}P_{k+1}^{\text{H}}\tilde{B}_{k} + \\ & A_{k} + H_{k}
            \end{aligned}
             & \text{otherwise}.
		\end{cases} 
    \end{align}
    The terminal conditions for the recursions~\eqref{eq:Val_quadcost_H} and~\eqref{eq:Val_quadcost_A} are:
    \begin{equation*}
        P_{L+1}^{\text{H}} := G_{L+1}^{\text{H}}, \quad P_{L+1}^{\text{A}} := G_{L+1}^{\text{A}}.
    \end{equation*} \frqed
\end{cor}

\begin{proof}{[Outline]}
    The proof directly follows from Theorem~\ref{th:NE_Val_gen_FDC}. Substituting the parameteric form of the saddle-point value~\eqref{eq:para_form_al_QC}, dynamics~\eqref{eq:linear_dynamics} and quadratic costs~\eqref{eq:cost_quad} in~\eqref{eq:TP_gen_H} and~\eqref{eq:TP_gen_A}, yields the policies~\eqref{eq:TP_quadcost_H} and~\eqref{eq:TP_quadcost_A}. Similar substitutions yield the saddle-point recursions~\eqref{eq:Val_quadcost_H} and~\eqref{eq:Val_quadcost_A}.
\end{proof}

\medskip

Corollary~\ref{cor:NE_Val_FDC_H} presents a closed-form solution for the \texttt{FlipDyn Co-op}~\eqref{eq:opti_E_cost} game with NE takeover strategies independent of state. The advantage of Corollary~\ref{cor:NE_Val_FDC_H} is the ability to compute saddle-point value and takeover policies as a function of state feedback at the expense of restricting the solution to a quadratic linear system. The recursions~\eqref{eq:Val_quadcost_H} and~\eqref{eq:Val_quadcost_A} are computed a priori, whereas, the policies~\eqref{eq:TP_gen_H} and~\eqref{eq:TP_gen_A} are computed using the state-feedback information. 

\subsubsection*{\underline{A Numerical Example}} We now evaluate the results of  Corollary~\ref{cor:NE_Val_FDC_H} on a discrete-time scalar linear time-invariant (LTI) system for a finite-horizon of $L = 30$. The LTI is characterized by the coefficients of dynamics~\eqref{eq:dynamics_HA_compact}, given by:
\begin{subequations}
\begin{align*}
    E_k &:= E = 1.0, \\
    B_k &:= B = -0.9\,\Delta t, \\
    C_k &:= C = -0.6\,\Delta t, \quad \forall k \in \mathcal{K}.
\end{align*}
\end{subequations}
The coefficients indicate that $\tilde{B} := \tilde{B}_{k} = (1-0.9\Delta t) < \tilde{C}_{k} := \tilde{C} = (1 - 0.6\Delta t)$, implying the control performance is better under a human agent as compared to an autonomous agent. We choose such a coefficient reflecting that humans can adapt and perform under challenging conditions. We choose $\Delta t = 0.1$ for the numerical example. However, such performance comes at increased costs compared to the autonomous agent, indicated by: 
\begin{equation*}
    \begin{aligned}
        G^{\text{H}}_k &:= G^{\text{H}} = 1.2, \ G^{\text{A}}_k := G^{\text{A}} = 1.0, \\  H_k &:= H = 0.35, \
        A_k := A = 0.2, \forall k \in \mathcal{K}.
    \end{aligned}
\end{equation*}
The state costs and takeover costs for the human agent is higher compared to that of the autonomous agent. 
We perform the numerical simulations with a fixed human intent probability $p_{k} = p, \forall k \in \mathcal{K}$, and determine the impact of varying the parameter $p$ on the underlying takeover policy. 

The saddle-point value parameters $P^{\text{H}}_{k}$ and $P^{\text{A}}_{k}$ for $p = 0.4, 0.55, 0.7$ are shown in Figures~\autoref{fig:SPV_p_04},~\autoref{fig:SPV_p_55}, and~\autoref{fig:SPV_p_7}, respectively. The corresponding takeover policies of both the agents are shown in Figures~\autoref{fig:TP_p_04},~\autoref{fig:TP_p_55}, and~\autoref{fig:TP_p_7}, respectively. When $p = 0.4$, we observe the saddle-point values of each agent cross over each others, with value of the autonomous agent higher at the start and flipping towards the end of the horizon. The impact of such a cross over saddle-point value is shown in Figure~\autoref{fig:TP_p_04}, where  both  agents choose to flip the state via a takeover when $\alpha = \text{H}$, and remain idle for the whole horizon. Next, for $p = 0.55$, we observe the saddle-point parameter of the autonomous agent jump between different values at the start of the horizon and reach a monotonic behavior at the later part of the horizon. Such a jumping nature is reflected in the takeover policy when $\alpha = \text{A}$, shown in Figure~\autoref{fig:TP_p_55}. Such a policy indicates to switch over the control to human agent, and once switched the control remains with the human, indicated when $\alpha = \text{H}$. However, towards the end of the horizon, both the autonomous agent and human agent request for the autonomous agent to takeover. Finally, when $p = 0.7$, the saddle-point value of both  agents are close to each other at the start of the horizon, which reflects the action of the human agent to take control of system when $\alpha = \text{A}$. The saddle-point value and takeover policies for the later part of the horizon remain the same. 

This numerical example illustrates the use of saddle-point value parameters in determining the takeover strategies for each player. Additionally, it provides insight into the system's behavior for the given costs and the system's stability properties. These insights are useful while designing the costs, which further impact the control and takeover policies.

\begin{figure*}[ht]
	\begin{center}
		\subfloat[]{\includegraphics[width = 0.32\linewidth]{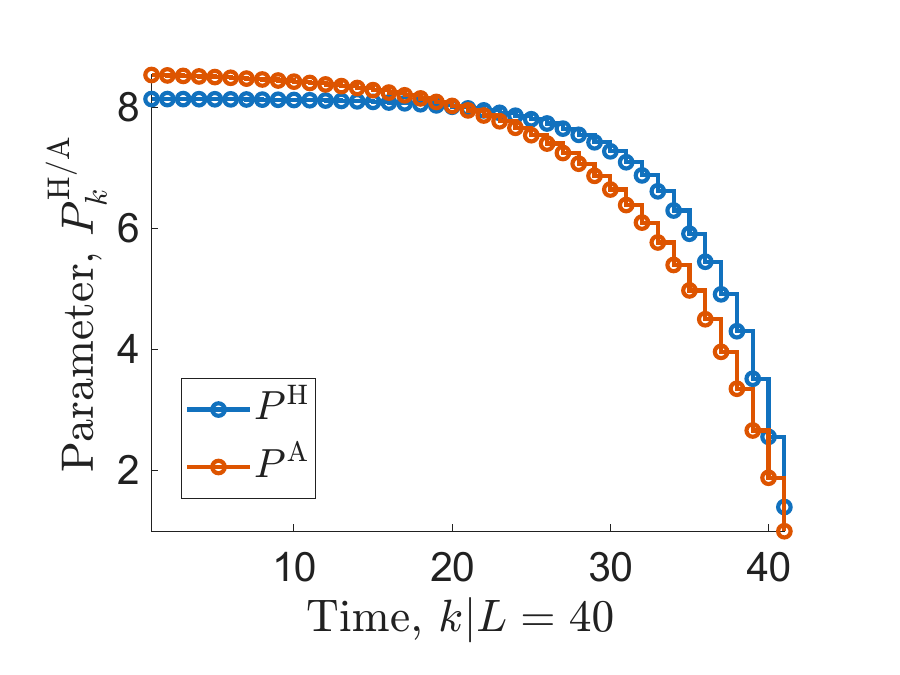}
			\label{fig:SPV_p_04}}
		\subfloat[]{\includegraphics[width = 0.32\linewidth]{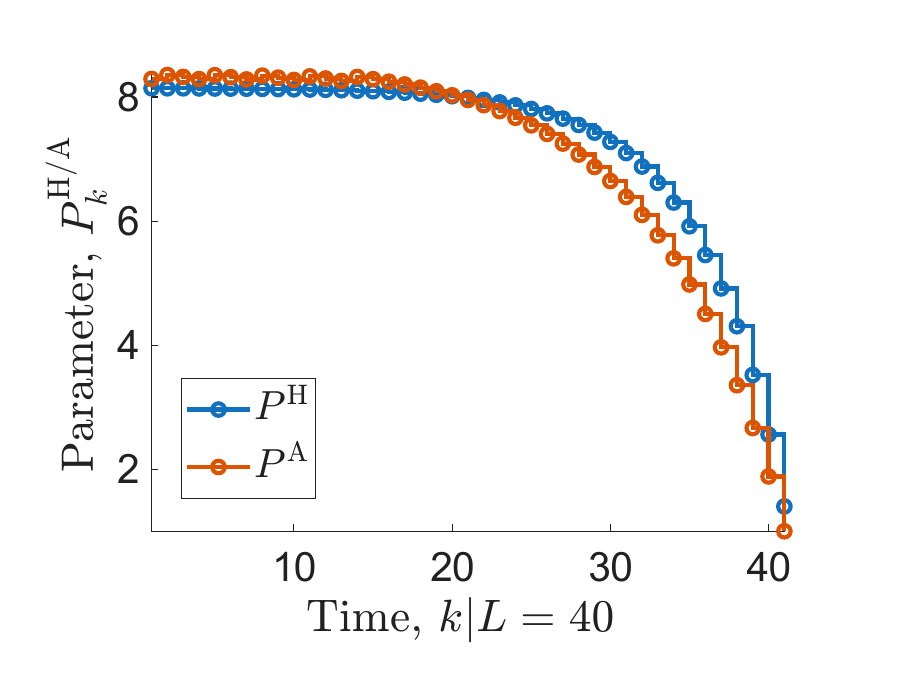}
			\label{fig:SPV_p_55}}
        \subfloat[]{\includegraphics[width = 0.32\linewidth]{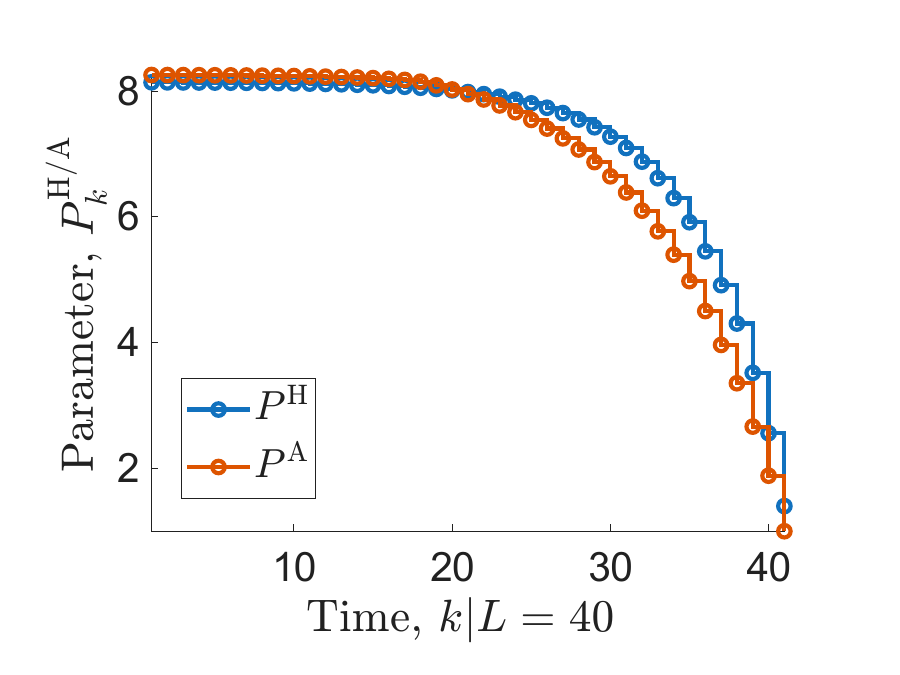}
			\label{fig:SPV_p_7}}
		\caption{\small Saddle-point value parameters for a scalar LTI with human intent probability a) $p = 0.4$, b) $p = 0.55$, and c)$p = 0.7$.}
	\end{center}
\end{figure*}

\begin{figure*}[ht]
	\begin{center}
		\subfloat[]{\includegraphics[width = 0.32\linewidth]{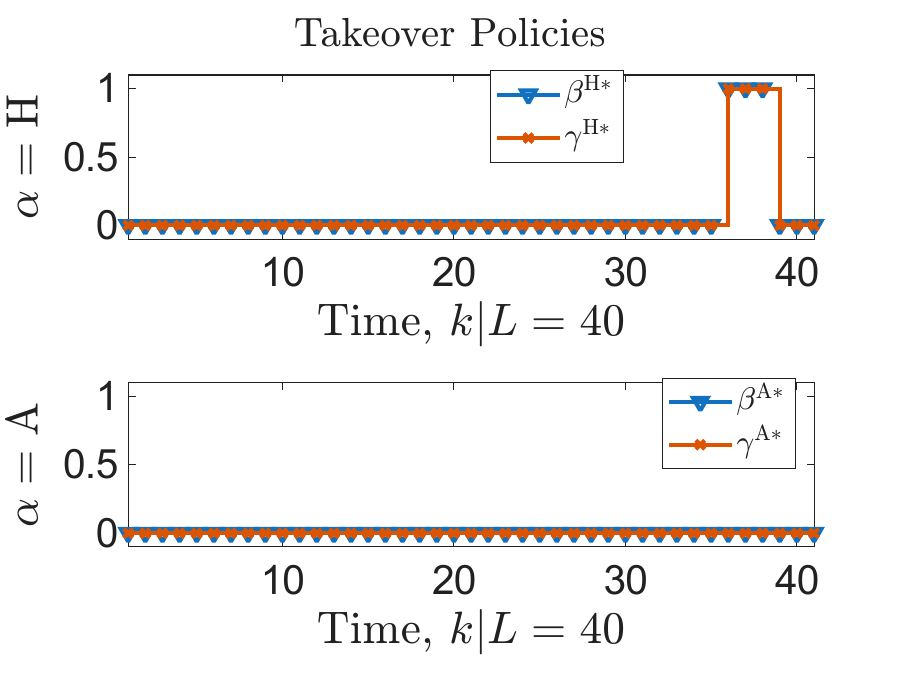}
			\label{fig:TP_p_04}}
		\subfloat[]{\includegraphics[width = 0.32\linewidth]{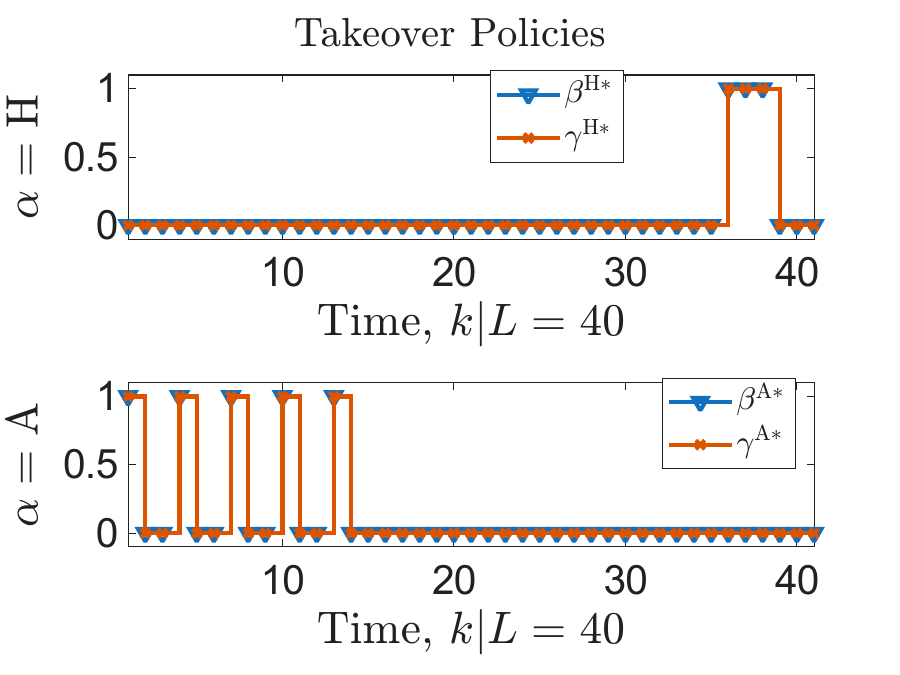}
			\label{fig:TP_p_55}}
        \subfloat[]{\includegraphics[width = 0.32\linewidth]{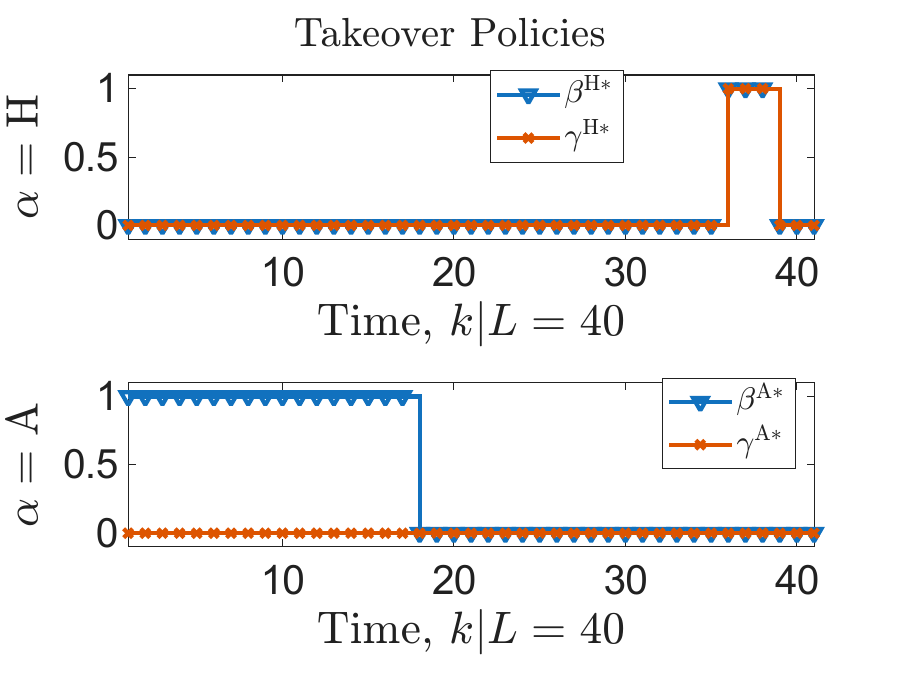}
			\label{fig:TP_p_7}}
		\caption{\small Takeover policy of the human ($\beta$) and autonomous agent LTI with human intent probability a) $p = 0.4$, b) $p = 0.55$, and c)$p = 0.7$.}
	\end{center}
\end{figure*}

\section{\texttt{Flip Co-Op} Potential LQ Game}\label{sec:FlipCoop_Potential_games}
The results derived in Theorem~\ref{th:NE_Val_gen_FDC} and Corollary~\ref{cor:NE_Val_FDC_H} are contingent on two core assumptions; i) knowing the underlying intent probability $p_{k}$ of the human agent and ii) identical utilities between the human and autonomous agent.
However, knowing the true intent probability $p_{k}$ is a challenge, as it cannot be measured directly and potentially a function of both state and time. Furthermore, we postulate, the probability $p_{k}$ is an artifact of mismatch between the utilities of the human and autonomous agent.
If the utilities of both the players were perfectly aligned, the underlying probability will collapse to the value of $1$. 
To alleviate such a limitation of the model, we will reformulate the problem as a potential game to solve the underlying shared takeover problem.

The state evolution under human and autonomous agent remain the same as defined in~\eqref{eq:human_dynamics} and~\eqref{eq:autonomous_dynamics}. The \texttt{FlipDyn} dynamics are given by~\eqref{eq:FlipDyn_compact} with the human agent intent probability $p_{k} = 1, \forall k \in \mathcal{K}$. Given a non-zero initial state $x_{1}$, the objective for human and autonomous agent are represented as:
\begin{equation}\label{eq:obj_non_zero_sum}
	\begin{aligned}
		J^{\tilde{\alpha}}(x_{1}, \alpha_{1}, \{\pi^{\text{H}}_{\mathbf{L}}\}, \{\pi^{\text{A}}_{\mathbf{L}}\}) = g_{L+1}^{\tilde{\alpha}}(x_{L+1}, \alpha_{L+1}) + \\ \sum_{t=1}^{L} g_t^{\tilde{\alpha}}(x_t, \alpha_t) + \pi_{t}^{\text{H}}h_t^{\tilde{\alpha}}(x_t) + \pi^{\text{A}}_{t}a_t^{\tilde{\alpha}}(x_t),
	\end{aligned}
\end{equation}
where $\tilde{\alpha} \in \{\text{H}^{\dag},\text{A}^{\dag}\}$ corresponds to either the human ($\text{H}^{\dag}$) or autonomous agent ($\text{A}^{\dag}$). Compared to~\eqref{eq:obj_def_alpha}, ~\eqref{eq:obj_non_zero_sum} indicates a distinct difference in the objective between the human and autonomous agent driven by the costs. 
We skip over the cost mappings in~\eqref{eq:obj_non_zero_sum}, as they are identical to~\eqref{eq:obj_def_alpha} for a given agent $\tilde{\alpha}$. 
We seek to find NE solution of the game~\eqref{eq:obj_non_zero_sum} in the same space of behavioral strategies, as defined in~\eqref{eq:behav_pol}:
\begin{equation*}
	y_{k}^{\alpha_{k}} = \begin{bmatrix}
		1 - \beta_{k}^{\alpha_{k}} & \beta_{k}^{\alpha_{k}}
	\end{bmatrix}^{\text{T}} \text{ and \hfill} z_{k}^{\alpha_{k}} = \begin{bmatrix}
	1 - \gamma_{k}^{\alpha_{k}} & \gamma_{k}^{\alpha_{k}}
\end{bmatrix}^{\text{T}},
\end{equation*}
such that $\{y_{\mathbf{L}}, z_{\mathbf{L}}\}$ correspond to the sequence of policy pairs over finite-horizon $L$. 
The expected outcome of the game~\eqref{eq:obj_non_zero_sum} for the human and autonomous agent is given as:
\begin{equation}\label{eq:E_obj_human_auto}
    \begin{aligned}
        & J_{E}^{\text{H}^{\dag}}(x_1, \alpha_{1}, y_{\mathbf{L}}, z_{\mathbf{L}}) :=  \mathbb{E}[J^{\text{H}^{\dag}}( x_{1}, \alpha_{1}, \{\pi^{\text{H}}_{L}\}, \{\pi^{\text{A}}_{L}\})], \\
        & J_{E}^{\text{A}^{\dag}}(x_1, \alpha_{1}, y_{\mathbf{L}}, z_{\mathbf{L}}) :=  \mathbb{E}[J^{\text{A}^{\dag}}( x_{1}, \alpha_{1}, \{\pi^{\text{H}}_{L}\}, \{\pi^{\text{A}}_{L}\})],
    \end{aligned}
\end{equation}
where the expectation is evaluated with respect to the policy pair distributions $\{y_{\mathbf{L}},z_{\mathbf{L}}\}$. 
For any non-zero initial state $x_{0}$ and $\alpha_{0}$, a pair policies $\{y_{\mathbf{L}}^{\bullet},z_{\mathbf{L}}^{\bullet}\}$ over a finite-horizon $L$ is said to be in NE if:
\begin{equation}\label{eq:Nash_outcome_pair}
    \begin{aligned}
         & J_E^{\text{H}^{\dag}}(x_0, \alpha_0, y_{\mathbf{L}}^{\bullet}, z_{\mathbf{L}}^{\bullet}) \leq J_E^{\text{H}^{\dag}}(x_0, \alpha_0, y_{\mathbf{L}}, z_{\mathbf{L}}^{\bullet}), \quad y_{\mathbf{L}} \in \Delta_{2}^{L} \\ 
         & J_E^{\text{A}^{\dag}}(x_0, \alpha_0, y_{\mathbf{L}}^{\bullet}, z_{\mathbf{L}}^{\bullet})
         \leq J_E^{\text{A}^{\dag}}(x_0, \alpha_0, y_{\mathbf{L}}^{\bullet}, z_{\mathbf{L}}), \quad z_{\mathbf{L}} \in \Delta_{2}^{L},
    \end{aligned}
\end{equation}
and $\left(J_E^{\text{H}^{\dag}}(x_0, \alpha_0, y_{\mathbf{L}}^{*}, z_{\mathbf{L}}^{*}), J_E^{\text{A}^{\dag}}(x_0, \alpha_0, y_{\mathbf{L}}^{*}, z_{\mathbf{L}}^{*})\right)$ is called the Nash outcome (NO) of the game. 
The \texttt{Flip Co-op} game in such a setting is defined by the cost pair~\eqref{eq:E_obj_human_auto} subject to the dynamics~\eqref{eq:human_dynamics},~\eqref{eq:autonomous_dynamics}, and~\eqref{eq:FlipDyn_compact}. 

Although NO~\eqref{eq:Nash_outcome_pair} of the game is defined for arbitrary costs and dynamics, we will restrict our solution to the class of linear dynamics~\eqref{eq:H_control_dynamics},~\eqref{eq:A_control_dynamics} under Assumption~\ref{ast:linear_control_space} and quadratic costs~\eqref{eq:cost_quad}. 

\subsection{Saddle-point value}
Analogous to~\eqref{eq:Cost_to_go_H} and~\eqref{eq:Cost_to_go_A}, we will define a cost-to-go matrix corresponding to the human and autonomous agent for a given \texttt{FlipDyn} state $\alpha_{k}$ at time $k$.

\underline{Case $\alpha_k = \text{H}$}: As the objective function~\eqref{eq:obj_non_zero_sum} differs between a human and autonomous agent, each \texttt{FlipDyn} state will contain two distinct cost-to-go matrix, $\Xi_{k}^{\text{H},\text{H}^{\dag}}$ and  $\Xi_{k}^{\text{H},\text{A}^{\dag}}$. These correspond to individual agent's perspective. The entries of the cost-to-go matrix $\Xi_{k+1}^{\text{H}, \tilde{\alpha}}$ for $\tilde{\alpha} = \{\text{H}^{\dag},\text{A}^{\dag}\}$ corresponding to each pair of takeover actions are given by:

\begin{equation}\label{eq:Cost_to_go_H_alpha}
    \begin{aligned}
		& \begin{matrix} & \hphantom{000} \text{Idle} & & \hphantom{v_{k+1}^0(.,.} \text{Takeover}\end{matrix} \\
		\begin{matrix} \text{Idle} \\[5pt] \text{Request to}\\\text{takeover} \end{matrix} & \underbrace{\begin{bmatrix}
			v_{k+1}^{\text{H},\tilde{\alpha}} &  v_{k+1}^{\text{H},\tilde{\alpha}} + a_k^{\tilde{\alpha}}(x)  \\[8pt]
			v_{k+1}^{\text{H},\tilde{\alpha}} + h_k^{\tilde{\alpha}}(x) &  v_{k+1}^{\text{A},\tilde{\alpha}} + h_k^{\tilde{\alpha}}(x) + a_k^{\tilde{\alpha}}(x) \\[5pt] 
		\end{bmatrix}}_{\Xi_{k+1}^{\text{H},\tilde{\alpha}}}
	\end{aligned},
\end{equation}
\begin{align}
    \label{eq:V_k^HHdag} \text{where } \ & v_{k+1}^{\text{H},\tilde{\alpha}} := V_{k+1}^{\text{H},\tilde{\alpha}}\left(F_k^{\text{H}}(x,u_k),\Xi_{k+2}^{\text{H},\tilde{\alpha}}\right), \\
    \label{eq:V_k^AHdag} & v_{k+1}^{\text{A},\tilde{\alpha}} := V_{k+1}^{\text{A},\tilde{\alpha}}\left(F_k^{\text{A}}(x,w_k),\Xi_{k+2}^{\text{A},\tilde{\alpha}}\right).
\end{align}
The row and column entries of $\Xi_{k+1}^{\text{H}, \tilde{\alpha}}$ are governed by \texttt{FlipDyn} dynamics~\eqref{eq:flip_state_H}. At time $k$, the saddle-point value in \texttt{FlipDyn} state $\alpha_{k} = \text{H}$ for $\tilde{\alpha} = \{\text{H}^{\dag},\text{A}^{\dag}\}$ satisfies: 
\begin{equation}
    \label{eq:V_H_saddle_point_alpha}
	V^{\text{H},\tilde{\alpha}}_k(x, \Xi_{k+1}^{\text{H},\tilde{\alpha}}) = g_k^{\tilde{\alpha}}(x,\text{H})  + \Val(\Xi^{\text{H},\tilde{\alpha}}_{k+1}). 
\end{equation}

\underline{Case $\alpha_k = \text{A}$}: Similar to the case $\alpha_{k} = \text{H}$, the entries of the cost-to-go matrix $\Xi_{k+1}^{\text{A}, \tilde{\alpha}}$ for $\tilde{\alpha} = \{\text{H}^{\dag},\text{A}^{\dag}\}$ corresponding to each pair of takeover actions are given by:
\begin{equation}\label{eq:Cost_to_go_A_alpha}
    \begin{aligned}
		& \begin{matrix} & \hphantom{0000} \text{Idle} & & \hphantom{v_{k+1}^0000} \text{Request to takeover}\end{matrix} \\
		\begin{matrix} \text{Idle} \\[8pt] \text{Takeover}\end{matrix} & \underbrace{\begin{bmatrix}
			v_{k+1}^{\text{A},\tilde{\alpha}} \hphantom{000}  &  \hphantom{00} v_{k+1}^{\text{A},\tilde{\alpha}} + a_k^{\tilde{\alpha}}(x) \hphantom{0}  \\[5pt]
			\begin{matrix}
                    v_{k+1}^{\text{H},\tilde{\alpha}}
                        + h_k^{\tilde{\alpha}}(x)       
                \end{matrix}
              &  \begin{matrix}
                    v_{k+1}^{\text{H},\tilde{\alpha}} + h_k^{\tilde{\alpha}}(x)
                        + a_k^{\tilde{\alpha}}(x)       
                \end{matrix}
		\end{bmatrix}}_{\Xi_{k+1}^{\text{A},\tilde{\alpha}}}.
	\end{aligned}
\end{equation}
Compared to~\eqref{eq:Cost_to_go_A}, in~\eqref{eq:Cost_to_go_A_alpha} the intent probability is $p_{k} = 1$, which results in a change in cost-to-go matrix. At time $k$, the saddle-point value in \texttt{FlipDyn} state $\alpha_{k} = \text{A}$ for $\tilde{\alpha} = \{\text{H}^{\dag},\text{A}^{\dag}\}$ satisfies: 
\begin{equation}
    \label{eq:V_A_saddle_point_alpha}
	V^{\text{A},\tilde{\alpha}}_k(x, \Xi_{k+1}^{\text{H},\tilde{\alpha}}) = g_k^{\tilde{\alpha}}(x,\text{A})  + \Val(\Xi^{\text{H},\tilde{\alpha}}_{k+1}). 
\end{equation}

\subsection{Bimatrix potential game}

While the saddle-point value functions~\eqref{eq:V_H_saddle_point_alpha} and~\eqref{eq:V_A_saddle_point_alpha} characterize the non-identical \texttt{Flip Co-Op} game~\eqref{eq:E_obj_human_auto} and the associated NE takeover policies, they come with notable challenges. In particular, solving the coupled dynamic programming equations for each player requires accounting four distinct saddle-point value functions and their associated cost-to-go matrices. Solving for each saddle-point value equates to solving a four-way coupled quadratic problem, with non-convex constraints yielding only a local minima.  

To alleviate these challenges, we propose reformulating the game as a \emph{potential bimatrix game}~\cite{mondererPotentialGames1996}. The core idea is to construct a single scalar function, \emph{potential function} whose changes under any unilateral deviation by either player exactly match the change in that player's objective. This allows the game to be analyzed using a unified potential cost-to-go matrix that captures the incentives of both agents in a common structure.

Given the FlipDyn state $\alpha_k \in \{\text{H}, \text{A}\}$, we define the \emph{potential game matrix} $\mathcal{P}_k^{\alpha_{k}}$ at time $k$ as:
\begin{equation}\label{eq:Potential_Game_Matrix}
    \mathcal{P}_{k+1}^{\alpha_{k}} :=
    \begin{bmatrix}
        \Psi_{k+1,11}^{\alpha_{k}} & \Psi_{k+1,12}^{\alpha_{k}}\\[5pt]
        \Psi_{k+1,21}^{\alpha_{k}} & \Psi_{k+1,22}^{\alpha_{k}}
    \end{bmatrix},
\end{equation}
where each entry $\Psi_{ij}^{\alpha_{k}}$ represents the cost-to-go of the system under the corresponding pair of takeover actions (row: human, column: autonomous agent). The $2 \times 2$ bimatrices $\Xi_{k}^{\text{H},\text{H}^{\dag}}$ and $\Xi_{k}^{\text{H},\text{A}^{\dag}}$ is an exact potential game, if there exists a potential $\mathcal{P}_k^{\text{H}}$, such that:

\begin{equation}\label{eq:potential_game_linear_eq}
    \begin{aligned}
        & \Psi_{k+1,11}^{\text{H}} - \Psi_{k+1,21}^{\text{H}} = v_{k+1}^{\text{H},\text{H}^{\dag}} - v_{k+1}^{\text{H},\text{H}^{\dag}} - h_{k}^{\text{H}^{\dag}}(x), \\
        & \Psi_{k+1,12}^{\text{H}} - \Psi_{k+1,22}^{\text{H}} = v_{k+1}^{\text{H},\text{H}^{\dag}} - v_{k+1}^{\text{A},\text{H}^{\dag}} - h_{k}^{\text{H}^{\dag}}(x), \\
        & \Psi_{k+1,11}^{\text{H}} - \Psi_{k+1,12}^{\text{H}} = v_{k+1}^{\text{H},\text{A}^{\dag}} - v_{k+1}^{\text{H},\text{A}^{\dag}} - a_{k}^{\text{A}^{\dag}}(x), \\
        & \Psi_{k+1,21}^{\text{H}} - \Psi_{k+1,22}^{\text{H}} = v_{k+1}^{\text{H},\text{A}^{\dag}} - v_{k+1}^{\text{A},\text{A}^{\dag}} - h_{k}^{\text{H}^{\dag}}(x).
    \end{aligned}
\end{equation}
The set of linear equations~\eqref{eq:potential_game_linear_eq} has a solution if and only if:
\begin{equation}
    \label{eq:existence_potential_func}
    \begin{aligned}
        & \Xi_{k}^{\alpha_{k},\text{H}^{\dag}}\text{\small{(1,1)}}   - \Xi_{k}^{\alpha_{k},\text{H}^{\dag}}\text{\small{(1,2)}} - \Xi_{k}^{\alpha_{k},\text{H}^{\dag}}\text{\small{(2,1)}} + \Xi_{k}^{\alpha_{k},\text{H}^{\dag}}\text{\small{(2,2)}}
        \\ & =
        \Xi_{k}^{\alpha_{k},\text{A}^{\dag}}\text{\small{(1,1)}} - \Xi_{k}^{\alpha_{k},\text{A}^{\dag}}\text{\small{(1,2)}} -  \Xi_{k}^{\alpha_{k},\text{A}^{\dag}}\text{\small{(2,1)}} + \Xi_{k}^{\alpha_{k},\text{A}^{\dag}}\text{\small{(2,2)}}, \\
        & \Rightarrow v_{k+1}^{\text{A},\text{H}^{\dag}} - v_{k+1}^{\text{H},\text{H}^{\dag}} = v_{k+1}^{\text{A},\text{A}^{\dag}} - v_{k+1}^{\text{H},\text{A}^{\dag}}.
    \end{aligned}
\end{equation}
The existence condition indicates that the difference in the saddle-point value between the hybrid states must be identical from both human and autonomous agent perspective.
A candidate solution~\cite{hespanha2017noncooperative} which satisfies~\eqref{eq:potential_game_linear_eq} is:
\begin{equation}
    \label{eq:potential_game_sol_H}
    \begin{aligned}
        & \Psi_{k+1,22}^{\text{H}} = 0,\\
        & \Psi_{k+1,12}^{\text{H}} = v_{k+1}^{\text{H},\text{H}^{\dag}} - v_{k+1}^{\text{A},\text{H}^{\dag}} - h^{\text{H}^{\dag}}_{k}(x), \\
        & \Psi_{k+1,21}^{\text{H}} = v_{k+1}^{\text{H},\text{H}^{\dag}} - v_{k+1}^{\text{A},\text{H}^{\dag}} - a^{\text{A}^{\dag}}_{k}(x), \\
        & \Psi_{k+1,11}^{\text{H}} = v_{k+1}^{\text{H},\text{H}^{\dag}} - v_{k+1}^{\text{A},\text{H}^{\dag}} - h^{\text{H}^{\dag}}_{k}(x) - a^{\text{A}^{\dag}}_{k}(x).
    \end{aligned}
\end{equation}
Similarly, for $\alpha_{k} = \text{A}$, the condition for existence~\eqref{eq:existence_potential_func} results in zero on both sides, which is always satisfied for any choice of cost.
The entries of a potential function matrix that yields an exact potential game is given by:
\begin{equation}
    \label{eq:potential_game_sol_A}
    \begin{aligned}
        & \Psi_{k+1,22}^{\text{A}} = 0,\\
        & \Psi_{k+1,12}^{\text{A}} = v_{k+1}^{\text{A},\text{H}^{\dag}} - v_{k+1}^{\text{H},\text{H}^{\dag}} - h^{\text{H}^{\dag}}_{k}(x), \\
        & \Psi_{k+1,21}^{\text{A}} = - a^{\text{A}^{\dag}}_{k}(x), \\
        & \Psi_{k+1,11}^{\text{A}} = v_{k+1}^{\text{A},\text{H}^{\dag}} - v_{k+1}^{\text{H},\text{H}^{\dag}} - h^{\text{H}^{\dag}}_{k}(x)-  a^{\text{A}^{\dag}}_{k}(x).
    \end{aligned}
\end{equation}
With the potential game~\eqref{eq:potential_game_sol_H} and~\eqref{eq:potential_game_sol_A} established for each of the \texttt{FlipDyn} state, we will characterize the NE takeover strategies and saddle-point values over the time-horizon $L$.

\subsection{Potential game admissible NE takeover strategies}
The stage and takeover quadratic costs are given by:
\begin{gather}\label{eq:cost_quad_PG}
        g_k(x,\alpha_k)^{\tilde{\alpha}} = x^{\tp}G_k^{\alpha_k,\tilde{\alpha}}x, \\
        h_k(x) = x^{\tp}H^{\tilde{\alpha}}_kx, \ a_k(x) = x^{\tp}A^{\tilde{\alpha}}_kx, \ \tilde{\alpha} \in \{\text{H}^{\dag},\text{A}^{\dag}\}. \nonumber
\end{gather}
Since we consider linear quadratic dynamics~\eqref{eq:H_control_dynamics},~\eqref{eq:A_control_dynamics} with quadratic costs~\eqref{eq:cost_quad_PG} and state feedback control law~\eqref{eq:linear_FD_control} (Assumption~\ref{ast:linear_control_space}), we postulate the same parametric form for the saddle-point value in each of \texttt{FlipDyn} state as:
\begin{equation}
    \label{eq:para_potential_SPV}
    V^{\text{H},\tilde{\alpha}}_{k}(x) := x^{\tp}Q_{k}^{\text{H},\tilde{\alpha}}x, \ V^{\text{A},\tilde{\alpha}}_{k}(x) := x^{\tp}Q_{k}^{\text{A},\tilde{\alpha}}x, \ \tilde{\alpha} := \{\text{H}^{\dag}, \text{A}^{\dag}\},
\end{equation}
where $Q_{k}^{\text{H},\tilde{\alpha}} \in \mathbb{S}_{+}^{n \times n}$ and $Q_{k}^{\text{A},\tilde{\alpha}} \in \mathbb{S}_{+}^{n \times n}$ for the associated \texttt{FlipDyn} state and player.
% Although each \texttt{FlipDyn} state is associated with two non-identical objectives corresponding to the human and autonomous player, the common potential game matrix~\eqref{eq:potential_game_sol_H} and~\eqref{eq:potential_game_sol_A} enables us to consider a \emph{common} saddle-point value, represented by~\eqref{eq:para_potential_SPV}.

Next, we outline the admissible NE takeover strategies of both the players, along with the corresponding saddle-point values for each of the \texttt{FlipDyn} state.

\begin{cor}\label{cor:Adm_NE_Val_FDC_H}
    (Case $\alpha_k = \text{H}$) The non-identical \texttt{Flip Co-op} game~\eqref{eq:E_obj_human_auto} governed by linear dynamics~\eqref{eq:dynamics_HA_compact} and \texttt{FlipDyn} dynamics~\eqref{eq:FlipDyn_compact} with quadratic costs~\eqref{eq:cost_quad}, admits a admissible NE takeover strategies at time $k \in \mathcal{K}$, given by:
    \begin{align}
    \begin{split}\label{eq:Adm_TP_quadcost_H}
            \{\beta^{\text{H}\bullet}_{k}, \gamma^{\text{H}\bullet}_{k}\} = \begin{cases}
            \{0,0\}, & \text{if } x^{\tp}Q^{\text{H},\text{H}^{\dag}}_{k+1}x < \begin{matrix}
                x^{\tp}Q^{\text{A},\text{H}^{\dag}}_{k+1}x + \\  x^{\tp}(H^{\text{H}^{\dag}}_{k} + A^{\text{A}^{\dag}}_{k})x,
            \end{matrix} \\[5 pt]
            \{1,1\}, & \text{otherwise } 
            \end{cases} 
    \end{split}
    \end{align}
    \begin{align}\label{eq:Adm_Val_quadcost_H}
        Q_{k}^{\text{H},\tilde{\alpha}} = 
        \begin{cases}
            \begin{aligned}
				& G_k^{\text{H},\tilde{\alpha}} + Q_{k+1}^{\text{H},\tilde{\alpha}}, 
            \end{aligned} &\text{if } Q^{\text{H},\text{H}^{\dag}}_{k+1} \prec \begin{matrix}
                Q^{\text{A},\text{H}^{\dag}}_{k+1} + H^{\text{H}^{\dag}}_{k} \\ + A^{\text{A}^{\dag}}_{k},
            \end{matrix} \\
            \begin{aligned}
				& G_k^{\text{H},\tilde{\alpha}} + Q_{k+1}^{\text{A},\tilde{\alpha}} \\ & + H_{k}^{\tilde{\alpha}} + A_{k}^{\tilde{\alpha}},
            \end{aligned} &\text{otherwise }.
		\end{cases} 
    \end{align}
    (Case $\alpha_k = \text{A}$) The unique NE takeover strategies are given by:
    \begin{align}
    \begin{split}\label{eq:Adm_TP_quadcost_A}
            \{\beta_{k}^{\text{A}\bullet},\gamma_{k}^{\text{A}\bullet}\}  = \begin{cases}
            \{0,0\}, & \text{if } \ 
                \begin{matrix}
                    x^{\tp}(Q^{\text{A},\text{H}^{\dag}}_{k+1} - Q^{\text{H},\text{H}^{\dag}}_{k+1})x < \\ x^{\tp}H^{\text{H}^{\dag}}_{k}x
                \end{matrix}, \\[10pt]
            \{1,0\}, & \text{otherwise. }
            \end{cases} 
    \end{split}
    \end{align}
    The saddle-point value is given by:
    \begin{align}\label{eq:Adm_Val_quadcost_A}
        Q_{k}^{\text{A},\tilde{\alpha}} = 
        \begin{cases}
            \begin{aligned}
				& G_k^{\text{A}, \tilde{\alpha}} + Q_{k+1}^{\text{A},\tilde{\alpha}},
            \end{aligned} &\text{if } \begin{matrix}
                    x^{\tp}(Q^{\text{A},\text{H}^{\dag}}_{k+1} - Q^{\text{H},\text{H}^{\dag}}_{k+1})x < \\ x^{\tp}H^{\text{H}^{\dag}}_{k}x
                \end{matrix}, \\[10pt]
            \begin{aligned}
				& G_k^{\text{A},\tilde{\alpha}} + Q_{k+1}^{\text{H},\tilde{\alpha}} \\ & + H_{k}^{\tilde{\alpha}},
            \end{aligned} &\text{otherwise }
		\end{cases} 
    \end{align}
    The terminal conditions for the recursions~\eqref{eq:Val_quadcost_H} and~\eqref{eq:Val_quadcost_A} are:
    \begin{equation*}
        Q_{L+1}^{\text{H},\tilde{\alpha}} := G_{L+1}^{\text{H},\tilde{\alpha}}, \quad Q_{L+1}^{\text{A},\tilde{\alpha}} := G_{L+1}^{\text{A},\tilde{\alpha}}.
    \end{equation*} \frqed
\end{cor}

\begin{proof}
    We will prove the corollary for \texttt{FlipDyn} state $\alpha_{k} = \text{H}$ and omit the proof for $\alpha_{k} = \text{A}$, as they are analogous. To determine the NE takeover policies, we will analyze the potential game matrix $\mathcal{P}_{k+1}^{\text{H}}$. Since the costs are quadratic~\eqref{eq:cost_quad}, comparing the entry $\Psi_{11}^{\text{H}}$ with respect to entries of the same row and column, yields the following condition:
    \begin{equation*}
    \begin{aligned}
        \Psi_{k+1,11}^{\text{H}} \leq \Psi_{k+1,11}^{\text{H}} - a_{k}^{\text{A}^{\dag}}(x), \\ 
        \Psi_{k+1,11}^{\text{H}} \leq \Psi_{k+1,11}^{\text{H}} - h_{k}^{\text{A}^{\dag}}(x).
    \end{aligned}
    \end{equation*}
    Such a condition, indicates that if the condition
    \begin{equation*}
        \Psi_{k+1,11}^{\text{H}} < 0
    \end{equation*}
    is satisfied, then it corresponds to the pure admissible NE policy of playing idle; else, it corresponds to the pure admissible NE policy of \{takeover, request to takeover\} by the autonomous and human agent, respectively, summarized in~\eqref{eq:Adm_TP_quadcost_H}. For the corresponding NE policy, the saddle-point value for each agent is associated with the matrices, $\Xi_{k+1}^{\text{H}^{\dag}}$ and $\Xi_{k+1}^{\text{A}^{\dag}}$ and not $\Psi_{11}^{\text{H}}$, described in~\eqref{eq:Adm_Val_quadcost_H}. 
\end{proof}

\section{Application to Tracking}\label{sec:Application}
We now demonstrate the applicability of the \texttt{Flip Co-op} framework on a trajectory tracking problem involving shared autonomy. 
Specifically, we consider a discrete-time linear system representing a vehicle navigating a path comprising a straight segment, a left arc (turn), and another straight segment, illustrated in Figure~\ref{fig:Illustration_application}. The system can be controlled either by a human or an autonomous agent, with each control mode associated with distinct cost structures and control effectiveness. The objective is to determine, over a finite-horizon, the optimal takeover strategy: i.e., when the human should assume control versus when the autonomous agent should. This setup allows us to evaluate how the theoretical saddle-point strategies developed under the LQ formulation translate to practical scenarios involving cost-performance trade-offs and intent-driven human-autonomy switching.

\begin{figure}[ht]
	\begin{center}
		\subfloat[]{\includegraphics[width = 0.55\linewidth]{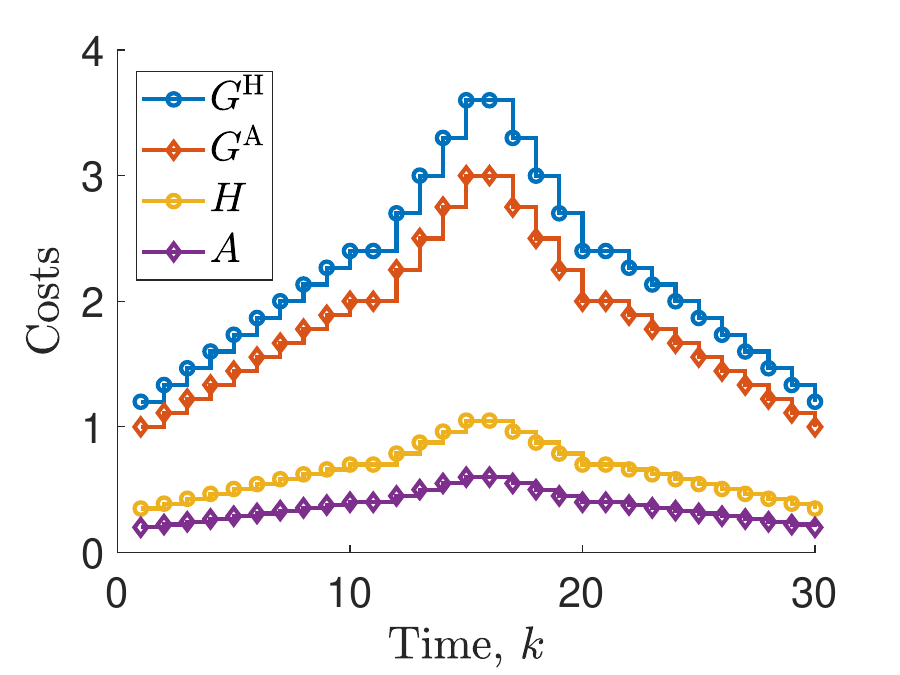}
			\label{fig:LQT_costs}}
		\subfloat[]{\includegraphics[width = 0.35\linewidth]{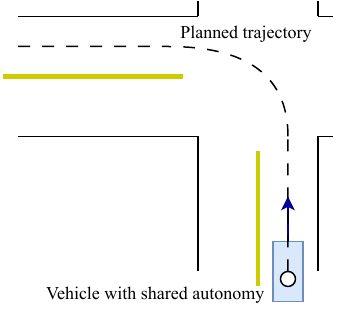}
			\label{fig:Illustration_application}}
		\caption{\small a) The state and takeover costs associated with the application. b) Illustration of the task for vehicle with shared autonomy. The dotted line indicates the desired trajectory.}
	\end{center}
\end{figure}

\subsection*{\texttt{Flip Co-op} LQ game}
The state and takeover costs corresponding to the path are illustrated in Figure~\ref{fig:LQT_costs}. 
The rising and falling patterns reflect increasing costs along the initial straight segment, peaking during the curved segment, and decreasing again as the vehicle returns to a straight path. 
We assign higher costs to the human agent compared to the autonomous agent.
For both agents, we assume dynamics of the form:
\[
x_{k+1} = (E + BK_{k})x_{k}, \quad  x_{k+1} = (E + CW_{k})x_{k}, k \in \mathcal{K}. 
\]
Given the path-dependent costs and arbitrary control matrices ($B$ and $C := wB, w < 1$), we solve for the control gains $K_{k}$ and $W_{k}$ using a standard backward dynamic programming method. 
These gains allow us to use the compact dynamics form~\eqref{eq:dynamics_HA_compact} as referenced in Corollary~\ref{cor:NE_Val_FDC_H}. 
The human intent probability $p_{k}$ is assumed to increase from 0.1 to 0.8 during the straight-to-curved transition, then decrease back to 0.1 in the final straight segment. 
Using Corollary~\ref{cor:NE_Val_FDC_H}, we solve for the takeover strategies and associated saddle-point values. We present results for two scenarios: (a) the autonomous agent controls the vehicle during a significant portion of the horizon, including the curved segment, and
(b) the human agent maintains control over most of the trajectory. 
The corresponding saddle-point values are shown in Figures~\ref{fig:SPV_p_LQT_case_a} and~\ref{fig:SPV_p_LQT_case_b}.

\begin{figure*}[ht]
	\begin{center}
		\subfloat[]{\includegraphics[width = 0.24\linewidth]{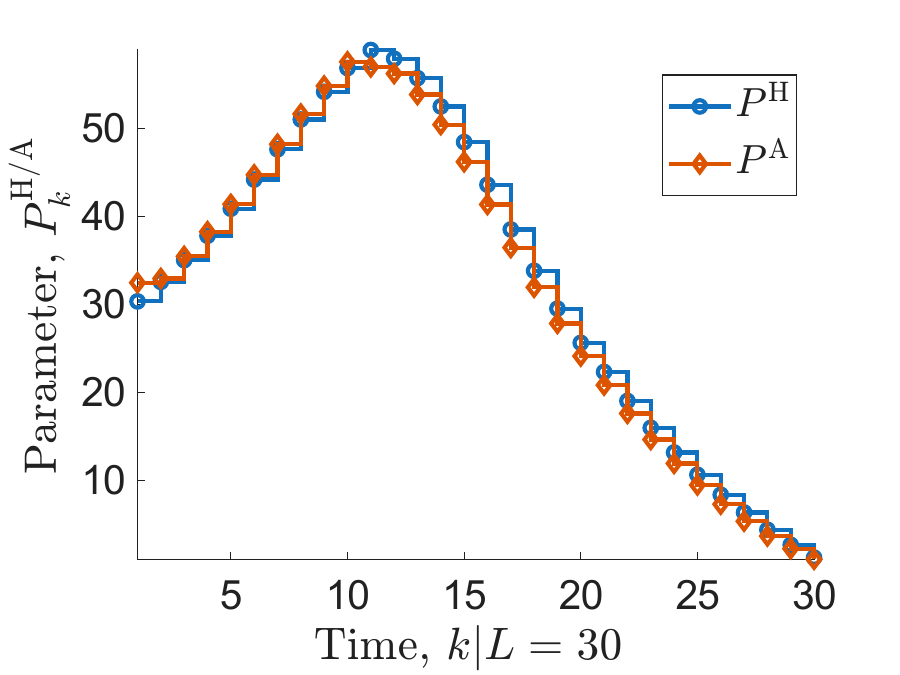}
			\label{fig:SPV_p_LQT_case_a}}
        \subfloat[]{\includegraphics[width = 0.24\linewidth]{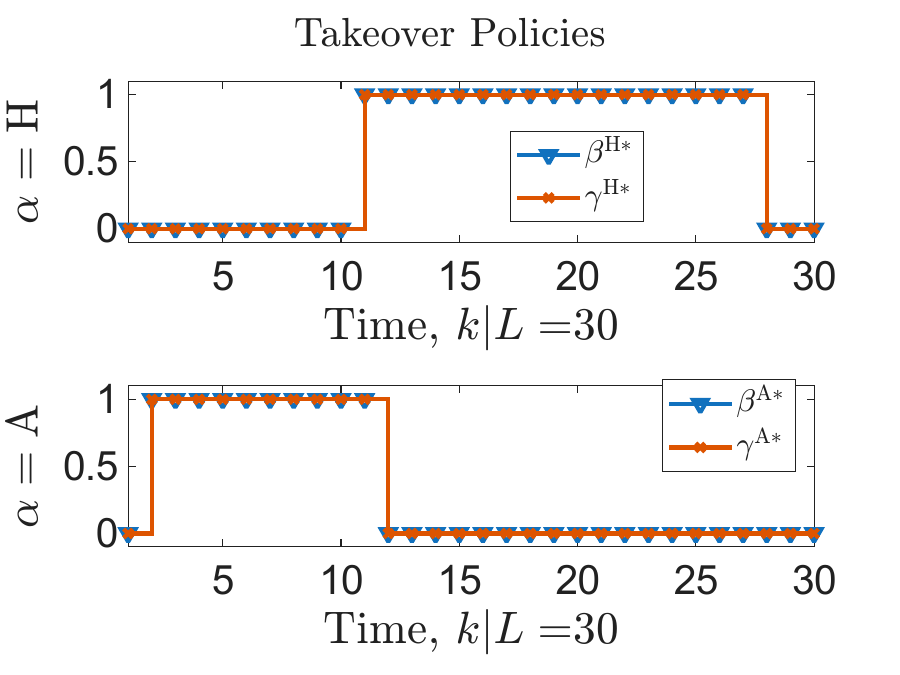}
			\label{fig:takeover_policy_plot_p_case_a}}
		\subfloat[]{\includegraphics[width = 0.24\linewidth]{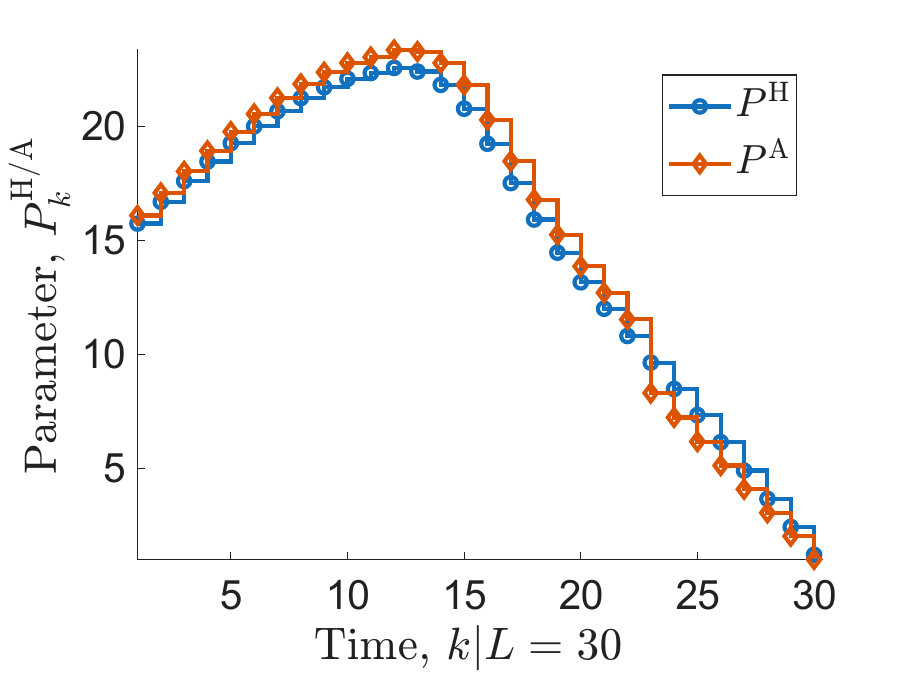}
			\label{fig:SPV_p_LQT_case_b}}
        \subfloat[]{\includegraphics[width = 0.24\linewidth]{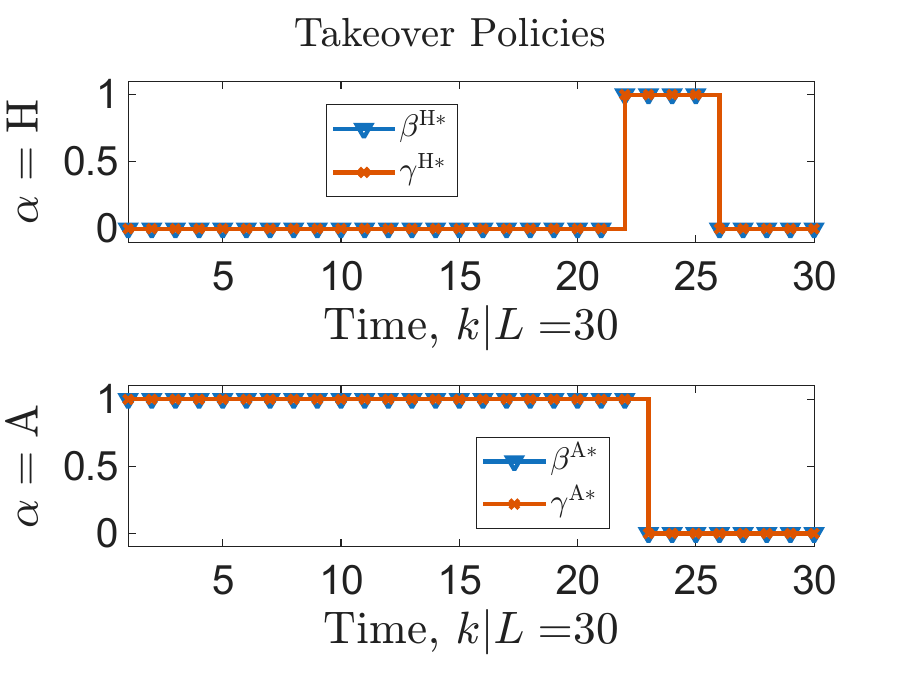}
			\label{fig:takeover_policy_plot_p_case_b}}
		\caption{\small a) The state and takeover costs associated with the application. b) Illustration of the task for vehicle with shared autonomy. The dotted line indicates the desired trajectory.}
	\end{center}
\end{figure*}
\underline{Case a:} The takeover policy when the autonomous agent controls the vehicle for a significant portion of the trajectory, including the curved segment, is shown in Figure~\ref{fig:takeover_policy_plot_p_case_a}. When the \texttt{FlipDyn} state is $\alpha = \text{H}$, the human retains control during the first straight segment ($\{\beta^{H*}, \gamma^{H*}\} = \{0,0\}$), followed by a request from the autonomous agent to take over in the curved segment, and ending with the human regaining control. In contrast, when $\alpha = \text{A}$, the human initially requests a takeover in the first straight segment, while no takeover requests are made during the curved and final straight segments.

\underline{Case b:} A similar analysis is shown in Figure~\ref{fig:takeover_policy_plot_p_case_b}. Here, when the vehicle starts under human control ($\alpha = \text{H}$), it remains so for most of the horizon, with a request from the autonomous agent to take over near the end. Conversely, when the vehicle starts under autonomous control ($\alpha = \text{A}$), the human requests a takeover almost immediately and maintains control for the majority of the trajectory.

These results highlight the effectiveness and flexibility of the proposed framework in reasoning about agent takeovers in shared autonomy. The ability to dynamically assign control based on evolving costs, intent, and dynamics illustrates the value of saddle-point strategies in hybrid decision-making environments.

\section{Conclusion and Future Directions}\label{sec:Conclusion}
This paper developed \texttt{Flip Co-op}, a cooperative game-theoretic framework for takeover in shared autonomy.
The framework yields Nash equilibrium takeover strategies that go beyond heuristic blending or switching. 
We established the existence and characterization of equilibria under stochastic human intent and derived closed-form recursions for linear–quadratic systems, enabling efficient computation of equilibrium policies in continuous state spaces. 
A potential-game reformulation further extended the model to partially misaligned utilities, ensuring tractability while capturing deviations in human intent. 
Application to vehicle trajectory tracking demonstrated how equilibrium takeover strategies adapt across path-dependent costs and intent probabilities, highlighting trade-offs between human adaptability and autonomous efficiency.

Future work will focus on advancing both theory and deployment. A key direction is to develop methods for learning the saddle-point value functions from data, enabling adaptive approximation of equilibrium policies in high-dimensional or nonlinear systems where closed-form recursions are intractable. 
Online inference of dynamic human intent, coupled with reinforcement or inverse game-theoretic learning, will further align takeover strategies with observed human behavior. 
Experimental validation on robotic platforms will be essential to assess robustness under sensing noise, delays, and workload variations. 
Finally, extensions to multi-human or networked autonomy scenarios will allow \texttt{Flip Co-op} to address cooperative takeovers at scale, moving toward safety-assured shared autonomy in complex cyber–physical systems.

\section{References}

\bibliographystyle{IEEEtran}
\bibliography{ref_arxiv}

\begin{thebibliography}{10}
\providecommand{\url}[1]{#1}
\csname url@rmstyle\endcsname
\providecommand{\newblock}{\relax}
\providecommand{\bibinfo}[2]{#2}
\providecommand\BIBentrySTDinterwordspacing{\spaceskip=0pt\relax}
\providecommand\BIBentryALTinterwordstretchfactor{4}
\providecommand\BIBentryALTinterwordspacing{\spaceskip=\fontdimen2\font plus
\BIBentryALTinterwordstretchfactor\fontdimen3\font minus
  \fontdimen4\font\relax}
\providecommand\BIBforeignlanguage[2]{{%
\expandafter\ifx\csname l@#1\endcsname\relax
\typeout{** WARNING: IEEEtran.bst: No hyphenation pattern has been}%
\typeout{** loaded for the language `#1'. Using the pattern for}%
\typeout{** the default language instead.}%
\else
\language=\csname l@#1\endcsname
\fi
#2}}

\bibitem{javdaniSharedAutonomyHindsight2015}
S.~Javdani, S.~S. Srinivasa, and J.~A. Bagnell, ``Shared {{Autonomy}} via
  {{Hindsight Optimization}},'' \emph{Robotics science and systems : online
  proceedings}, vol. 2015, p. 10.15607/RSS.2015.XI.032, July 2015.

\bibitem{nikolaidisHumanRobotMutualAdaptation2017}
S.~Nikolaidis, Y.~X. Zhu, D.~Hsu, and S.~Srinivasa, ``Human-{{Robot Mutual
  Adaptation}} in {{Shared Autonomy}},'' in \emph{Proceedings of the 2017
  {{ACM}}/{{IEEE International Conference}} on {{Human-Robot
  Interaction}}}.\hskip 1em plus 0.5em minus 0.4em\relax Vienna Austria: ACM,
  Mar. 2017, pp. 294--302.

\bibitem{jeonSharedAutonomyLearned2020}
H.~J. Jeon, D.~P. Losey, and D.~Sadigh, ``Shared {{Autonomy}} with {{Learned
  Latent Actions}},'' May 2020.

\bibitem{reddySharedAutonomyDeep2018}
S.~Reddy, A.~D. Dragan, and S.~Levine, ``Shared {{Autonomy}} via {{Deep
  Reinforcement Learning}},'' May 2018.

\bibitem{liDifferentialGameTheory2019}
Y.~Li, G.~Carboni, F.~Gonzalez, D.~Campolo, and E.~Burdet, ``Differential game
  theory for versatile physical human--robot interaction,'' \emph{Nature
  Machine Intelligence}, vol.~1, no.~1, pp. 36--43, Jan. 2019.

\bibitem{sadighPlanningAutonomousCars2016}
D.~Sadigh, S.~Sastry, S.~A.~Seshia, and A.~D.~Dragan, ``Planning for
  {{Autonomous Cars}} that {{Leverage Effects}} on {{Human Actions}},'' in
  \emph{Robotics: {{Science}} and {{Systems XII}}}.\hskip 1em plus 0.5em minus
  0.4em\relax {Robotics: Science and Systems Foundation}, 2016.

\bibitem{nikolaidisGameTheoreticModelingHuman2017}
S.~Nikolaidis, S.~Nath, A.~D. Procaccia, and S.~Srinivasa, ``Game-{{Theoretic
  Modeling}} of {{Human Adaptation}} in {{Human-Robot Collaboration}},'' in
  \emph{Proceedings of the 2017 {{ACM}}/{{IEEE International Conference}} on
  {{Human-Robot Interaction}}}, ser. {{HRI}} '17.\hskip 1em plus 0.5em minus
  0.4em\relax New York, NY, USA: Association for Computing Machinery, Mar.
  2017, pp. 323--331.

\bibitem{draganPolicyblendingFormalismShared2013}
A.~D. Dragan and S.~S. Srinivasa, ``A policy-blending formalism for shared
  control,'' \emph{The International Journal of Robotics Research}, vol.~32,
  no.~7, pp. 790--805, June 2013.

\bibitem{scerriAdjustableAutonomyReal2002}
P.~Scerri, D.~V. Pynadath, and M.~Tambe, ``Towards {{Adjustable Autonomy}} for
  the {{Real World}},'' \emph{Journal of Artificial Intelligence Research},
  vol.~17, pp. 171--228, Sept. 2002.

\bibitem{loseyPhysicalInteractionCommunication2022}
\BIBentryALTinterwordspacing
D.~P. Losey, A.~Bajcsy, M.~K. O’Malley, and A.~D. Dragan, ``Physical
  interaction as communication: {{Learning}} robot objectives online from human
  corrections,'' vol.~41, no.~1, pp. 20--44. [Online]. Available:
  \url{https://doi.org/10.1177/02783649211050958}
\BIBentrySTDinterwordspacing

\bibitem{banikFlipDynGameResource2022a}
S.~Banik and S.~D. Bopardikar, ``{{FlipDyn}}: {{A}} game of resource takeovers
  in dynamical systems,'' in \emph{2022 {{IEEE}} 61st {{Conference}} on
  {{Decision}} and {{Control}} ({{CDC}})}, Dec. 2022, pp. 2506--2511.

\bibitem{vandijkFlipItGameStealthy2013}
M.~{van Dijk}, A.~Juels, A.~Oprea, and R.~L. Rivest, ``{{FlipIt}}: {{The Game}}
  of ``{{Stealthy Takeover}}'','' \emph{Journal of Cryptology}, vol.~26, no.~4,
  pp. 655--713, Oct. 2013.

\bibitem{zhu2015game}
Q.~Zhu and T.~Basar, ``Game-theoretic methods for robustness, security, and
  resilience of cyberphysical control systems: games-in-games principle for
  optimal cross-layer resilient control systems,'' \emph{IEEE {C}ontrol Systems
  Magazine}, vol.~35, no.~1, pp. 46--65, 2015.

\bibitem{banikFlipDynGraphsResource2025}
S.~Banik, S.~D. Bopardikar, and N.~Hovakimyan, ``{{FlipDyn}} in~{{Graphs}}:
  {{Resource Takeover Games}} in~{{Graphs}},'' in \emph{Decision and {{Game
  Theory}} for {{Security}}}, A.~Sinha, J.~Fu, Q.~Zhu, and T.~Zhang, Eds.\hskip
  1em plus 0.5em minus 0.4em\relax Cham: Springer Nature Switzerland, 2025, pp.
  220--239.

\bibitem{hespanha2017noncooperative}
J.~P. Hespanha, \emph{Noncooperative game theory: An introduction for engineers
  and computer scientists}.\hskip 1em plus 0.5em minus 0.4em\relax Princeton
  University Press, 2017.

\bibitem{mondererPotentialGames1996}
D.~Monderer and L.~S. Shapley, ``Potential {{Games}},'' \emph{Games and
  Economic Behavior}, vol.~14, no.~1, pp. 124--143, May 1996.

\end{thebibliography}

\end{document}